\documentclass[12pt]{article}
\usepackage{ownstyle}
\author{Abolfazl Motahari, Guy Bresler and David Tse\\
Department of Electrical Engineering and Computer Sciences\\
University of California, Berkeley\\
\{motahari,gbresler,dtse\}@eecs.berkeley.edu}
\date{}
\begin{document}
\title{Information Theory of DNA Shotgun Sequencing}
\maketitle

\begin{abstract}
DNA sequencing is the basic workhorse of modern day biology and medicine. Shotgun sequencing is the dominant technique used: many randomly located short fragments called reads are extracted from the DNA sequence, and these reads are assembled to reconstruct the original sequence. A basic question is: given a sequencing technology and the statistics of the DNA sequence, what is the minimum number of reads required for reliable reconstruction? This number provides a fundamental limit to the performance of {\em any} assembly algorithm.  For a simple statistical model of the DNA sequence and the read process, we show that the answer admits a critical phenomena in the asymptotic limit of long DNA sequences: if the read length is below a  threshold, reconstruction is impossible no matter how many reads are observed, and if the read length is above the threshold, having enough reads to cover the DNA sequence is sufficient to reconstruct. The threshold is computed in terms of the Renyi entropy rate of the DNA sequence. We also study the impact of noise in the read process on the performance.
\end{abstract}

\section{Introduction}
\label{sec:intro}

\subsection{Background and Motivation}

DNA sequencing is the basic workhorse of modern day biology and medicine. Since the sequencing of the Human Reference Genome ten years ago, there has been an explosive advance in sequencing technology, resulting in several orders of magnitude increase in throughput and decrease in cost. This advance allows the generation of a massive amount of data, enabling the exploration of a diverse set of questions in biology and medicine that were beyond reach even several years ago. These questions include discovering genetic variations across different humans (such as single-nucleotide polymorphisms SNPs), identifying genes affected by mutation in cancer tissue genomes, sequencing an individual's genome for diagnosis (personal genomics), and understanding DNA regulation in different body tissues.

Shotgun sequencing is the dominant method currently used to sequence long strands of DNA, including entire genomes. The basic shotgun DNA sequencing set-up is shown in Figure~\ref{fig:setup}. Starting with a DNA molecule, the goal is to obtain the sequence of nucleotides ($A,C,G$ or $T$) comprising it. (For humans, the DNA sequence has about $3 \times 10^9$ nucleotides, or base pairs.)  The sequencing machine extracts a large number of reads from the DNA; each read is a randomly located fragment of the DNA sequence, of lengths of the order of 100-1000 base pairs, depending on the sequencing technology. The number of reads can be of the order of 10's to 100's of millions. The {\em DNA assembly problem} is to reconstruct the DNA sequence from the many reads.

\begin{figure}
\begin{center}
\includegraphics[width=5.5in]{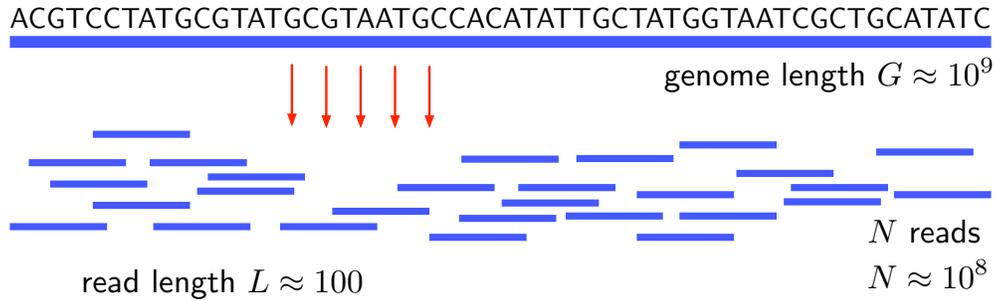}
\end{center}\vspace{-0.5cm}
\caption{Schematic for shotgun sequencing.}
\label{fig:setup}
\end{figure}

When the human genome was sequenced in 2001, there was only one sequencing technology, the Sanger platform \cite{FSA77}. Since 2005, there has been a proliferation of ``next generation" platforms, including Roche/454, Life Technologies SOLiD, Illumina Hi-Seq 2000 and Pacific Biosciences RS.
Compared to the Sanger platform, these technologies can provide massively parallel sequencing, producing far more reads per instrument run and at a lower cost, although the reads are shorter in lengths.
Each of these technologies generates reads of different lengths and with different noise profiles. For example, the 454 machines have read lengths of about 400 base pairs, while the SOLiD machines have read lengths of about $100$ base pairs. At the same time, there has been a proliferation of a large number of assembly algorithms, many tailored to specific sequencing technologies. (Recent survey articles \cite{Pop09,MKS10,PS10} discuss no less than $20$ such algorithms, and the Wikipedia entry on this topic listed $42$ \cite{wiki:assembly}.)

The design of these algorithms is based primarily on  {\em computational} considerations. The goal is to design efficient algorithms that can scale well with the large amount of sequencing data. Current  algorithms are often tailored to particular machines and are designed based on heuristics and domain knowledge regarding the specific DNA being sequenced; this makes it difficult to compare different algorithms, not to mention to define what it means by an ``optimal" assembly algorithm for a given sequencing problem. One reason for the heuristic approach taken towards the problem is that various formulations of the assembly problem are known to be NP-hard (see for example \cite{KS05})

An alternative to the computational view is the {\em information theoretic} view. In this view, the genome sequence is regarded as a random string to be estimated based on the read data. The basic question is: what is the minimum number of reads needed to reconstruct the DNA sequence with a given reliability? This minimum number can be used as a benchmark to compare different algorithms, and an optimal algorithm is one that achieves this minimum number. It can also provide an algorithm-independent basis for comparing different sequencing technologies and for designing new technologies.

This information theoretic view falls in the realm of {\em DNA sequencing theory} \cite{wiki:dna_theory}. A well-known lower bound on the number of reads needed can be obtained by a {\em coverage analysis}, an approach pioneered by Lander and Waterman \cite{LW88}. This lower bound is the number of reads $\Ncov$ such that with a desired probability, say $1-\epsilon$, the randomly located reads cover the entire genome sequence. The number $\Ncov$  can be easily approximated:
$$ \Ncov(\eps,G,L) \approx \frac{G}{L} \ln \left (\frac{G}{L\eps} \right ),$$
where $G$ and $L$ are DNA and read length, respectively. While this is clearly a lower bound on the minimum number of reads needed, it is in general not tight: only requiring the reads to cover the entire genome sequence does not guarantee that consecutive reads can actually be stitched back together to recover the entire sequence. The ability to do that depends on other factors such as the repeat statistics of the DNA sequence and also the noise profile in the read process. Thus, characterizing the minimum number of reads required for reconstruction is in general an open question.

\subsection{Main contributions}

In this paper, we make progress on this basic problem. We first focus on a very simple model:
\begin{enumerate}
\item the DNA sequence is modeled as an i.i.d. random process of length $G$ with each symbol taking values according to a probability distribution ${\bf p}$ on the alphabet $\{A,C,G,T\}$.
\item each read is of length $L$ symbols and begins at a uniformly distributed location on the DNA sequence and the locations are independent from one read to another.
\item the read process is noiseless.
\end{enumerate}

Fix an $\eps  \in (0,1/2)$ and let $\Nmin (\eps,G,L)$  be the minimum number of reads required to reconstruct the DNA with probability at least $1-\eps$. We would like to know how $\frac{\Nmin (\eps,G,L)}{\Ncov (\eps,G,L)}$ behaves in the asymptotic regime when $G$ and $L$ grow to infinity. It turns out that in this regime, the ratio depends on $G$ and $L$ through a normalized parameter:
$$ \bar{L} := \frac{L}{\log G},$$
where $\log(\cdot)$ represents logarithms to base 2. We define
$$\capacity (\bar L)= \lim_{G \rightarrow \infty,L = \bar L \log G } \frac{\Nmin(\eps,G,L)}{\Ncov(\eps,G,L)}.$$

Let $H_2({\bf p})$ be the Renyi entropy of order 2, defined to be
\begin{equation}
\label{eq:Renyi}
H_2({\bf p }) := - \log \sum_i p_i^2.
\end{equation}

Our main result, Theorem \ref{t:mainCapacity},  yields a {\em critical phenomenon}: when $\bar{L}$ is below the threshold $2/H_2({\bf p})$, reconstruction is impossible, i.e. $\capacity (\bar{L}) = \infty$, but when $\bar{L}$ is above that threshold,  the obvious necessary condition of coverage is also sufficient for reconstruction, i.e. $\capacity (\bar{L}) = 1$. The significance of the threshold is that when $\bar{L} < 2/H_2({\bf p})$, with high probability there are many repeats of length $L$ in the DNA sequence, while when  $\bar{L} > 2/H_2({\bf p})$, with high probability there are no repeats of length $L$. Thus, another way to interpret the result is that $\bar{L} < 2/H_2({\bf p})$  is a  {\em repeat-limited} regime while $\bar{L} > 2/H_2({\bf p})$ is a {\em coverage-limited} regime.
The result is summarized in Figure~\ref{fig:plot}.

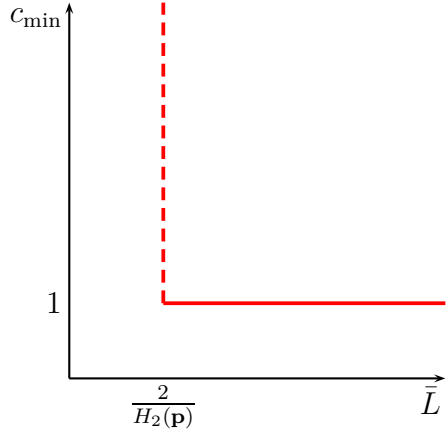
\begin{figure}\centering
\begin{pspicture}(-1,-1)(5,5)
\psline{->}(0,0)(5,0)
\psline{->}(0,0)(0,5)
\rput[t](4.8,-.1){$\bar L $}
\psline[linecolor=red,linestyle=dashed,linewidth=.05](1.25,1)(1.25,5)
\psline[linecolor=red,linewidth=.05](1.25,1)(5,1)
\rput[r](-.1,4.8){$\capacity$}
\rput[t](1.25,-.1){$2\over H_2(\mathbf{p})$}
\rput[r](-.1,1){1}
\end{pspicture}
\vspace{-5mm}
\caption{The critical phenomenon.}
\label{fig:plot}
\end{figure}

A standard measure of data requirements in DNA sequencing projects is the {\em coverage depth}:  the average number of reads covering each base pair. Thus, $\Ncov (\eps, G,L) \times  L/G$ is the coverage depth required to cover the DNA sequence with probability $1-\eps$ (as predicted by Lander-Waterman), and $\Nmin (\eps,G,L) \times L/G$ is the minimum coverage depth required to reconstruct the DNA sequence with probability $1-\eps$. Hence,
$\capacity(\bar{L})$ can be interpreted as the (asymptotic) normalized minimum coverage depth required to reconstruct the DNA sequence.

In a related work, Arratia et al \cite{Arr96} showed that $\bar{L} > 2/H_2({\bf p})$ is a necessary and sufficient condition for reconstruction of the i.i.d. DNA sequence if {\em all} length $L$ subsequences of the DNA sequence are given as reads. This arises in a technology called {\em sequencing by hybridization}.
Obviously, for the same read length $L$, having all length $L$ subsequences provides more information than any number of reads from shotgun sequencing, where the reads are randomly sampled. Hence, it follows that $\bar{L} > 2/H_2({\bf p})$ is also a necessary condition for shotgun sequencing. What our result says is that this condition together with coverage is sufficient for reconstruction asymptotically.

The basic model of i.i.d. DNA sequence and noiseless reads is very simplistic.  We provide two extensions to our basic result: 1) Markov DNA statistics; 2) noisy reads. In the first case, we show that the same result as the i.i.d. case holds except that the Renyi entropy $H_2({\bf p})$ is replaced by the Renyi entropy rate of the Markov process. In the second case, we analyze the performance of a modification of the greedy algorithm to deal with noisy reads, and show that the effect of noise is on increasing the threshold on the read length below which reconstruction is impossible.

Even with these extensions, our models still miss several important aspects of real DNA and read data.  Perhaps the most important aspect is the presence of long repeats in the DNA sequences of many organisms, ranging from bacteria to humans. These long repeats are poorly captured by i.i.d. or even Markov models due to their short-range correlation. Another aspect is the non-uniformity of the sampling of reads from the DNA sequences.  At the end of the paper, we will discuss how our results can be used as a foundation to tackle these and other issues.

\subsection{Related Work}

Li \cite{Min90} has also posed the question of minimum number of reads for the i.i.d. equiprobable DNA sequence model. He showed that if $L > 4 \log G$, then the number of reads needed is $O(G/L \ln G ) $, i.e. a constant multiple of the number needed for coverage. Specializing to the equiprobable case, our result shows that reconstruction is possible with probability $1-\eps$ {\em if and only if} $L > \log G$ and the number of reads is $G/L \ln (G/L\eps)$. Not only is our characterization necessary and sufficient, we have a much weaker condition on the read length $L$, and we get the correct pre-log constant on the number of reads needed.
As will be seen later, many different algorithms have the same scaling behavior in the number of reads they need, but it is the pre-log constant which distinguishes them.

A common formulation of  DNA assembly is the shortest common superstring (SCS) problem. The SCS problem is the problem of finding the shortest string containing a set of strings, where in the DNA assembly context, the given strings are the reads and the superstring is the estimate of the original DNA sequence.  While the general SCS problem with arbitrary instances is NP-hard \cite{KS05}, the greedy algorithm has been shown to be optimal for the SCS problem under certain probabilistic settings \cite{Frieze:kx, Ma09}.  Thus, the reader may have the impression that our results overlap with these previous works. However, there are significant differences.

First, at a basic problem formulation level, the SCS problem and the DNA sequence reconstruction problem are not equivalent: there is no guarantee that the shortest common superstring containing the given reads is the original DNA sequence. Indeed, it has already been observed in the assembly literature (eg. \cite{MB09}) that the shortest common superstring of the reads may be a significant compression of the original DNA sequence, especially when the latter has a lot of repeats, since finding the shortest common superstring tends to merge these repeats.  For example, in the case of very short reads the resulting shortest common superstring is definitely not the original DNA sequence. In contrast, we formulate the problem directly in terms of reconstructing the original sequence, and a lower bound on the required read length emerges as part of the result.

Second, even if we assume that the shortest common superstring containing the reads is the original DNA sequence, one cannot recover our result from either \cite{Ma09} or \cite{Frieze:kx}, for different reasons.  The main result (Theorem 1) in \cite{Ma09} says that if one models the DNA sequence as an arbitrary sequence perturbed by mutating each symbol independently with probability $p$ and the reads are arbitrarily located, the average length of the sequence output by the greedy algorithm is no more than a factor of $1+3\delta$ of the length of the shortest common superstring, provided that $p > 2 \log (GL)/(\delta L)$, i.e. $ p > 2/(\delta \bar{L})$.  However, since $p \le 1$, the condition on $p$ in their theorem implies that $\delta \ge \frac{2}{\bar{L}}$.  Thus, for a fixed $\bar{L}$ they actually only showed that the greedy algorithm is approximately optimal to within a factor of $1+ 6/\bar{L}$, and optimal only under the further condition that $\bar{L} \rightarrow \infty$. In contrast, our result shows that the greedy algorithm is  optimal for {\em any} $\bar{L} > 2/H_2({\bf p})$, albeit under a weaker model for the DNA sequence (i.i.d. or Markov) and read locations (uniform random).



Regarding \cite{Frieze:kx}, the probabilistic model they used does not capture the essence of the DNA sequencing problem.   In their model, the given reads are all independently distributed and not from a single ``mother" sequence, as in our model. In contrast, in our model, even though the original DNA sequence is assumed to be i.i.d., the reads will be highly {\em correlated}, since many of the reads will be physically overlapping. In fact, it follows from \cite{Frieze:kx} that, given $N$ reads and the read length $L$ scaling like $\log N$, the length of the shortest common superstring scales like $N \log N$. On the other hand, in our model, the length of the reconstructed sequence would be proportional to $N$. Hence, the length of the shortest common superstring is much longer for the model studied in \cite{Frieze:kx}, a consequence of the reads being independent and therefore much harder to merge. So the two problems are completely different, although coincidentally the greedy algorithm is optimal for both problems.

\subsection{Notations and Outline}

A brief remark on notation is in order. Sets (and probabilistic events) are denoted by calligraphic type, e.g. $\calA, \BB, \EE$, vectors by boldface, e.g. $\bs,\bx,\by$, and random variables by capital letters such as $S,X,Y$. Random vectors are denoted by capital boldface, such as $\s, \bX, \bY$. The exception to these rules, for the sake of consistency with the literature, are the (non-random) parameters $G, N,$ and $L$. The natural logarithm is denoted by $\ln(\, \cdot \, )$ and the base~2 logarithm by $\log(\, \cdot \, )$.

The rest of the paper is organized as follows. Section~\ref{sec:formulation} gives the precise formulation of the problem.  Section~\ref{sec:upper} explains why reconstruction is impossible for read length below the stated threshold. For read length above the threshold,  an optimal algorithm is presented in Section~\ref{sec:achieve}, where a heuristic argument is given to explain why it performs optimally.  Sections~\ref{sec:source} and \ref{sec:noise} describe extensions of our basic result to incorporate read noise and a more complex model for DNA statistics, respectively.  Section ~\ref{sec:discussions} discusses future work. Appendices contain the formal proofs of all the results in the paper.

\section{I.i.d. DNA Model}

This section states the main result of this paper, addressing the optimal assembly of i.i.d. DNA sequences. We first  formulate the problem and state the result. Next, we compare the performance of the optimal algorithm with that of other existing algorithms. Finally, we discuss the computational complexity of the algorithm.

\subsection{Formulation and Result}
\label{sec:formulation}
The DNA sequence  $\bs=s_1 s_2 \dots s_G$ is modeled as an i.i.d. random process of length $G$ with each symbol taking values according to a probability distribution $\p=(p_1,p_2,p_3,p_4)$ on the alphabet $\{A,C,G,T\}$. For notational convenience we instead denote the letters by numerals, i.e. $s_i\in \{1,2,3,4\}$. To avoid boundary effects, we assume that the DNA sequence is circular, i.e., $s_i=s_j$ if $i=j$ mod $G$; this simplifies the exposition, and all results apply with appropriate minor modification to the non-circular case as well.


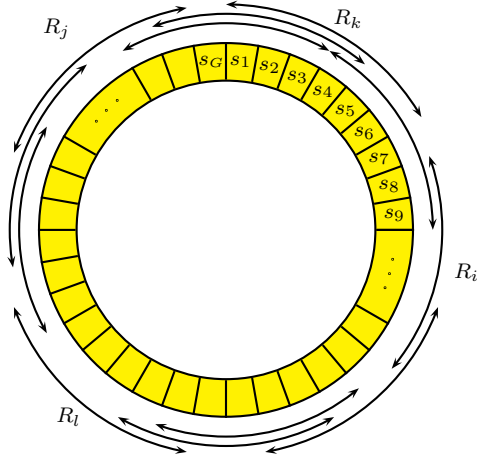
\begin{figure}
\centering
\scalebox{1}
{
\begin{pspicture}(-3,-3)(3,3)
\pscircle[fillstyle=solid,fillcolor=yellow]{2.5}
\pscircle[fillstyle=solid,fillcolor=white]{2}
\multido{\i=90+-10,\ij=85+-10,\ik=1+1}{9}{\psline(2;\i)(2.5;\i)\rput(2.25;\ij){\scriptsize $s_{\ik}$}}
\psline(2;0)(2.5;0)
\psline(2;0)(2.5;0)
\multido{\i=100+10,\ij=95+10,\ik=1+-1}{3}{\psline(2;\i)(2.5;\i)}
\rput(2.25;95){\scriptsize $ s_{\scriptscriptstyle G }$}
\multido{\i=-10+-5}{3}{\pscircle(2.25;\i){.01}}
\multido{\i=130+5}{3}{\pscircle(2.25;\i){.01}}
\multido{\i=-30+-10}{19}{\psline(2;\i)(2.5;\i)}
\psarc{<->}{2.75}{0}{60}
\rput{0}(3.25;-10){\scriptsize $R_i$}
\rput{0}(3.25;60){\scriptsize $R_k$}
\rput{0}(3.5;130){\scriptsize $R_{j}$}
\rput{0}(3.25;-130){\scriptsize $R_l$}
\psarc{<->}{3}{30}{90}
\psarc{<->}{2.875}{50}{110}
\psarc{<->}{2.75}{60}{120}
\psarc{<->}{3}{100}{160}
\psarc{<->}{2.875}{130}{190}
\psarc{<->}{2.75}{150}{210}
\psarc{<->}{3}{200}{260}
\psarc{<->}{2.875}{240}{300}
\psarc{<->}{2.75}{250}{310}
\psarc{<->}{3}{280}{340}
\psarc{<->}{2.875}{320}{380}
\end{pspicture}
}
\caption{A circular DNA sequence which is sampled randomly.}\label{f:circular DNA}
\end{figure}

The objective of DNA sequencing is to reconstruct the whole sequence $\bs$ based on $N$ \emph{reads} drawn randomly from the sequence (see Figure \ref{f:circular DNA}).
A read is a substring of length $L$ from the DNA sequence. The set of reads is denoted by $\RR =\{\br_1,\br_2,\ldots,\br_{N}\}$. The starting location of read $i$ is $t_i$, so $r_i=\bs[t_i,t_i+L-1]$.  The set of starting locations of the reads is denoted $\TT=\{t_1,t_2,\ldots,t_{N}\}$, where we assume $1\leq t_1\leq t_2\leq \dots\leq t_{N}\leq G$. We also assume that the starting location of each read is uniformly distributed on the DNA and the locations are independent from one read to another.

An \emph{assembly algorithm} takes a set of $N$ reads $\RR=\{\br_{1},\dots,\br_N\}$ and returns an estimated sequence $\hat \bs=\hat\bs(\RR)$.
We require \emph{perfect reconstruction}, which presumes that the algorithm $\phi$ makes an error if $\hat\bs\neq \bs$\footnote{The notion of perfect  reconstruction can be thought of as a mathematical idealization of the notion of ``finishing" a sequencing project as defined by the National Human Genome Research Institute \cite{finishing}, where finishing a chromosome requires at least 95\% of the chromosome to be represented by a contiguous sequence.}. We let $\P$ denote the probability model for the (random) DNA sequence $\s$ and the sample locations $\TT$, and $\EE :=\{\hat\s\neq \s\}$ the error event. A question of central interest is: what are the conditions on the read length $L$  and the number of reads $N$ such that the reconstruction error probability is less than a given target $\epsilon$? 
Unfortunately, this is in general a difficult question to answer.
We instead ask an easier {\em asymptotic} question: what is the ratio of the minimum number of reads $\Nmin$ and number of reads needed to cover the sequence $\Ncov$ as $L,G \rightarrow \infty$ with $\bar L =L/\log G$ being a constant, and which algorithm achieves the optimal performance asymptotically? More specifically, we are interested in $\capacity (\bar L)$, which is defined as
\begin{equation}
\capacity (\bar L ) = \lim_{G \to \infty, L=\bar L \log G} \frac{\Nmin(\epsilon,G,L)}{\Ncov(\epsilon,G,L)}.
\end{equation}

The main result for this model is:
\begin{theorem}\label{t:mainCapacity}
Fix an $\epsilon < 1/2$. \Fr $\capacity(\bar L)$ is given by
\begin{equation}
\label{eq:capacity}
\capacity(\bar L)=
\begin{cases}
\infty & \text{if} ~ \bar{L}  < 2/H_2({\bf p}),\\
1 & \text{if} ~ \bar{L} > 2/H_2({\bf p}),
\end{cases}
\end{equation}
where $H_2(\p)$ is the Renyi entropy of order 2 defined in \eqref{eq:Renyi}.
\end{theorem}

Section \ref{sec:upper} proves the first part of the theorem, that reconstruction is impossible for $\bar L < 2/H_2({\bf p})$. Section \ref{sec:achieve} shows how a simple greedy algorithm can achieve optimality for $ \bar L > 2/H_2({\bf p})$.

\subsection{$\bar L < \frac{2}{H_2({\bf p})}$: Repeat-limited regime}

\label{sec:upper}

The random nature of the DNA sequence gives rise to a variety of patterns.
The key observation in \cite{ukkonen92} is that there exist two patterns in the DNA sequence precluding reconstruction from an arbitrary set of reads of length $L$. In other words, reconstruction is not possible even if the $L$-spectrum, i.e. the set of all substrings of length $L$ appearing in the DNA sequence,  is given. The first pattern is the three way repeat of a substring of length $L-1$. The second pattern is two interleaved pairs of repeats of length $L-1$, shown in Figure~\ref{fig:repeat}. Arratia et al. \cite{Arr96} carried out a thorough analysis of randomly occurring repeats for the same i.i.d. DNA model as ours, and showed that the second pattern of two iterleaved repeats is the typical event for reconstruction to fail. A consequence of Theorem~7 in \cite{Arr96} is the following lemma (see also \cite{DFS94}).


\begin{figure}
\begin{center}
\includegraphics[scale=.5]{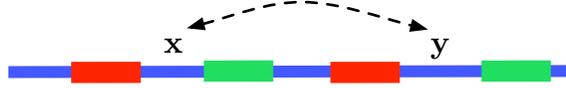}
\end{center}
\vspace{-.4cm}
\caption{Two pairs of interleaved repeats of length $L-1$ create ambiguity: from the reads it is impossible to know whether the sequences $\bx$ and $\by$ are as shown, or swapped. }
\label{fig:repeat}
\end{figure}

\begin{lemma}[Arratia et al. \cite{Arr96}]\label{t:repeat}
Fix $\bar{L} < \frac2{H_2(\mathbf{p})}$. An i.i.d. random DNA sequence contains interleaved repeats of length $L = \bar{L}\log G$ with probability $1-o(1)$.
\end{lemma}

We give a heuristic argument for the lemma, following \cite{Arr96}. As shown below, the expected number of length-$L$ repeats in an i.i.d. sequence has a rather sharp transition, from almost none to many, as $\bar L$ decreases below $\frac2{H_2(\mathbf{p})}$. It turns out that the positions of the repeats are approximately uniformly distributed throughout the sequence, so if there are many repeats, then it is likely that at least two of them are interleaved. 

We proceed with computing the expected number of repeats. Denoting by $\s_i^L$ the length-$L$ subsequence starting at position $i$, we have
 \begin{equation}\label{e:expectedDuplicates1}
 \E[\text{\# of length $L$ repeats}] = \sum_{1\leq i<j\leq G} \P(\s_i^L=\s_j^L)\,.
 \end{equation}
 Now, the probability that two specific \emph{physically disjoint} length-$\ell$ subsequences are identical is
 $$\Big(\sum_i p_i^2\Big)^\ell= e^{-\ell H_2(\p)}\,,$$
where $H_2(\p)=-\log \big(\sum_i p_i^2\big)$ is the R\'enyi entropy of order 2.
 Ignoring the $GL$ terms in \eqref{e:expectedDuplicates1} where $\s_i^L$ and $\s_j^L$ overlap, we get a lower bound,
\begin{equation}
\label{eq:expdupl}
\E[\text{\# of repeats}] >  \left(\frac{G^2}2 - GL\right) \,e^{-L H_2(\p)} \approx \frac{G^2}2 \,e^{-L H_2(\p)} .
\end{equation}
This number approaches zero if $\bar{L}> 2/H_2(\p)$ and it approaches infinity if $\bar{L}< 2/H_2(\p)$. 

Hence, if $\bar{L}< 2/H_2(\p)$, then the probability of having two pairs of interleaved repeats is very high. Moreover, as a consequence of Lemma \ref{lem self match iid} in Appendix \ref{sec:basic_proofs}, the contribution of the terms in (\ref{e:expectedDuplicates1}) due to physically overlapping subsequences is not large, and so the lower bound in (\ref{eq:expdupl}) is essentially tight. This suggests that  $\bar{L} = 2/H_2(\p)$ is in fact {\em the} threshold for existence of interleaved repeats.

\noindent
{\bf Proof of Theorem \ref{t:mainCapacity}, Part 1:}
\begin{proof}
The probability of observing a sequence of reads $\br_1, \ldots, \br_N$  given a DNA sequence $\bs$ is:
$$ \P(\br_1,\ldots, \br_N| \bs) = \prod_{i=1}^N \P(\br_i|\bs) = \prod_{i=1}^N \frac{\mbox{$\#$ of occurrences of $\br_i$ in $\bs$}}{G}.$$

Now suppose the DNA sequence $\bs$ has two interleaved repeats of length $L-1$ as in Figure~\ref{fig:repeat} and let $\bs'$ be the sequence with the subsequences $\bx$ and $\by$ swapped. Then,  the number of occurrences of each read $\br_i$ in $\bs$ and $\bs'$ is the same and hence
$$ \P(\br_1, \ldots, \br_N|\bs) = \P(\br_1, \ldots, \br_N|\bs').$$
Moreover, $\P(\bs) = \P(\bs').$ Hence
$$ \P(\bs|\br_1, \ldots, \br_N) = \P(\bs'|\br_1, \ldots, \br_N).$$
Thus, the optimal MAP rule will have a probability of reconstruction error of at least $1/2$ conditional on the DNA sequence having interleaved repeats of length $L-1$, regardless of the number of reads. By Lemma \ref{t:repeat}, this latter event has probability approaching $1$ as $G \rightarrow \infty$ if $\bar{L} < 2/H_2({\bf p})$. Since $\epsilon < 1/2$, this implies that for sufficiently large $G$, $\Nmin (\epsilon, G,L) = \infty$, thus proving the result.

\end{proof}

Note that for any {\em fixed} read length $L$, the probability of the interleaved repeat event will approach $1$ as the DNA length $G \rightarrow \infty$. This means that if we had defined \fr for a fixed read length $L$, then for {\em any} value of $L$ \fr would have been $\infty$. Thus, to get a meaningful result, one must scale $L$ with $G$, and Lemma \ref{t:repeat} suggests that letting $L$ and $G$ grow while fixing $\bar{L}$ is the correct scaling.



\subsection{$\bar L > \frac{2}{H_2({\bf p})}$: Coverage-limited regime}
\label{sec:achieve}

In order to reconstruct the DNA sequence it is necessary to observe each of the nucleotides, i.e. the reads must cover the sequence (see Figure~\ref{fig:coverage}). Worse than the missing nucleotides, a gap in coverage also creates ambiguity in the order of the contiguous pieces. Thus, $\Ncov(\epsilon, G,L)$, the minimum number of reads to cover the entire DNA sequence with probability $1-\epsilon$, is a lower bound to $\Nmin(\epsilon,G,L)$, the minimum number of reads to reconstruct with probability $1-\epsilon$.
The paper of Lander and Waterman \cite{LW88} studied the coverage problem in the context of DNA sequencing, and from their results, one can deduce the following asymptotics for $\Ncov(\epsilon, G,L)$.

\begin{figure}
\begin{center}
\vspace{-0.5cm}
\includegraphics[scale=.5]{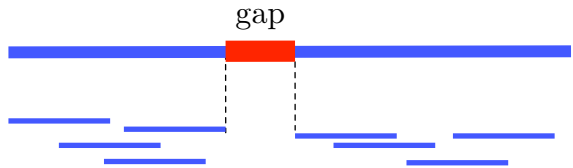}
\end{center}
\vspace{-.5cm}
\caption{The reads must cover the sequence.}
\label{fig:coverage}
\end{figure}

\begin{lemma}
\label{lem:cov_asym}
For any $\epsilon \in (0,1)$:
$$ \lim_{L,G \rightarrow \infty, L/\log G = \bar{L}} \frac{\Ncov(\eps, G,L)}{G/\bar{L}} = 1.$$
\end{lemma}

A standard coupon collector-style argument proves this lemma in \cite{LW88}. An intuitive justification of the lemma, which will be useful in the sequel, is as follows. To a very good approximation, the starting locations of the reads are given according to a Poisson process with rate $\lam=N/G$, and thus each spacing has an exponential$(\lam)$ distribution. Hence, the probability that there is a gap between two successive reads is approximately $e^{-\lam L}$. Hence, the expected number of gaps is approximately:
$$ N e^{-\lambda L}. $$
Asymptotically, this quantity is bounded away from zero if $N < G/\bar{L}$, and approaches zero otherwise.




We show that for $\bar{L} > 2/H_2(\bf p)$, a simple greedy algorithm (perhaps surprisingly) attains the coverage lower bound. Essentially, the greedy algorithm merges the reads repeatedly into {\em contigs}\footnote{Here, a contig means a contiguous fragment formed by overlapping sequenced reads.}, and the merging is done greedily, according to an overlap score defined on pairs of strings. For a given score the algorithm is as follows.

\paragraph{Greedy Algorithm:}

Input: $\RR$, the set of reads of length $L$.
\begin{enumerate}
\item Initialize the set of contigs as the given reads.
\item Find two contigs with largest overlap score, breaking ties arbitrarily, and merge them into one contig.
\item Repeat Step 2 until only one contig remains.
\end{enumerate}
For the i.i.d. DNA model and noiseless reads, we use the overlap score $W({\bf s_1},{\bf s_2}) $, defined as the length of the longest suffix of ${\bf s_1}$ identical to a prefix of ${\bf s_2}$.

Showing optimality of the greedy algorithm entails showing that if the reads cover the DNA sequence and there are no repeats of length $L$, then the greedy algorithm can reconstruct the DNA sequence. In the remainder of this subsection we heuristically explain the result, and we give a detailed proof in Appendix\ref{sec:basic_proofs}.

Since the greedy algorithm merges reads according to overlap score, we may think of the algorithm as working in stages, starting with an overlap score of $L$ down to an overlap  score of $0$. At stage $\ell$, the merging is between contigs with overlap score $\ell$. The key is to find the typical stage at which the {\em first} error in merging occurs. Assuming no errors have occurred in stages $L,L-1, \ldots, \ell+1$, consider the situation in stage $\ell$, as depicted in Figure~\ref{fig:contigs}. The algorithm has already merged the reads into a number of contigs. The boundary between two neighboring contigs is where the overlap between the neighboring reads is less than or equal to $\ell$; if it were larger than $\ell$, the two contigs would have been merged already. Hence, the expected number of contigs at stage $\ell$ is the expected number of pairs of successive reads with spacing greater than $L-\ell$. Again invoking the Poisson approximation, this is roughly equal to
$$Ne^{-\lambda(L-\ell)},$$
where $\lambda = N/G $.

\begin{figure}
\begin{center}
\includegraphics[width=5in]{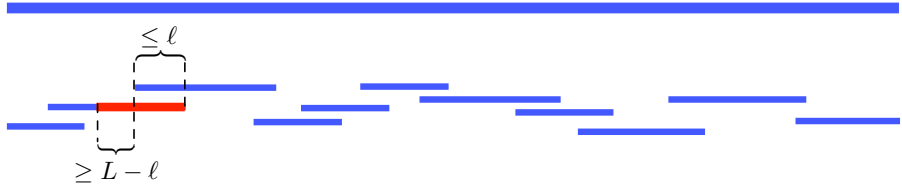}
\end{center}
\vspace{-.6cm}
\caption{The greedy algorithm merges reads into contigs according to the amount of overlap. At stage $\ell$ the algorithm has already merged all reads with overlap greater than $\ell$. The red segment denotes a read at the boundary of two contigs; the neighboring read must be offset by at least $L-\ell$.}
\label{fig:contigs}
\end{figure}

 Two contigs will be merged in error in stage $\ell$ if the length $\ell$ suffix of one contig equals the length $\ell$ prefix of another contig from a different location. Assuming these substrings are physically disjoint, the probability of this event is:
 $$ 2^{-\ell H_2({\bf p})}.$$
 Hence, the expected number of pairs of contigs for which this confusion event happens is approximately:
 \begin{equation}
 \label{eq:pairs}
  \left[Ne^{-\lambda(L-\ell)}\right]^2 \cdot 2^{-\ell H_2({\bf p})}\,.
 \end{equation}
 This number is largest either when $\ell = L$ or $\ell = 0$. This suggests that, typically, errors occurs in stage $L$ or stage $0$ of the algorithm.
 Errors occur at stage $L$ if there are repeats of length $L$ substrings in the DNA sequence. Errors occur at stage $0$ if there are still leftover unmerged contigs. The no-repeat condition ensures that the probability of the former event is small and the coverage condition ensures that the probability of the latter event is small. Hence, the two necessary conditions are also sufficient for reconstruction.

\subsection{Performance of Existing Algorithms}

The greedy algorithm was used by several of the most widely used genome assemblers for Sanger data, such as phrap, TIGR Assembler~\cite{Sut95} and CAP3~\cite{HM99}. More recent software aimed at assembling short-read sequencing data uses different algorithms. We will evaluate the normalized coverage depth of some of these algorithms on our basic statistical model and compare them to the information theoretic limit. The goal is not to compare between different algorithms; that would have been unfair since they are mainly designed for more complex scenarios including noisy reads and repeats in the DNA sequence. Rather, the aim is to illustrate our information theoretic framework and make some contact with existing assembly algorithm literature.

\subsubsection{Sequential Algorithm}

By merging reads with the largest overlap first, the greedy algorithm discussed above effectively grows the contigs {\em in parallel}. An alternative greedy strategy, used by software like
SSAKE \cite{War07}, VCAKE \cite{Jec07} and SHARCGS \cite{Doh07}, grows one contig {\em sequentially}. An unassembled read is chosen to start a contig, which is then repeatedly extended (say towards the right) by identifying reads that have the largest overlap with the contig until no more extension is possible. The algorithm succeeds if the final contig is the original DNA sequence.

The following proposition gives the normalized coverage depth of this algorithm.

\begin{proposition}\label{t:sequentialRate}
\Fr for the sequential algorithm is $\rate_{\text{seq}} (\bar L) = {\bar LH_2(\p) \ln 2\over \bar L H_2(\p) - 1}$ if $\bar{L} > 2/H_2(\mathbf{p})$.
\end{proposition}


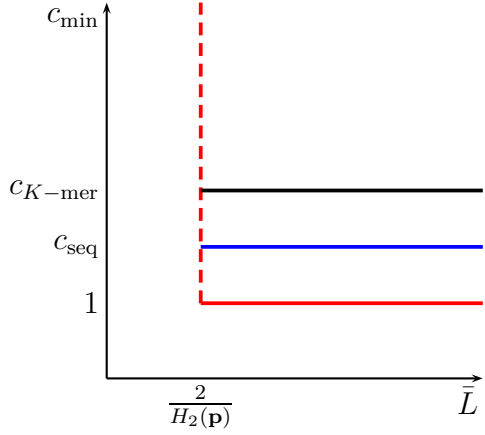
\begin{figure}\centering
\begin{pspicture}(-1,-1)(5,5)
\psline{->}(0,0)(5,0)
\psline{->}(0,0)(0,5)
\rput[t](4.8,-.1){$\bar L $}
\psline[linecolor=red,linestyle=dashed,linewidth=.05](1.25,1)(1.25,5)
\psline[linecolor=red,linewidth=.05](1.25,1)(5,1)
\psline[linecolor=blue,linewidth=.05](1.25,1.75)(5,1.75)
\psline[linecolor=black,linewidth=.05](1.25,2.5)(5,2.5)
\rput[r](-.1,4.8){$\capacity$}
\rput[t](1.25,-.1){$2\over H_2(\mathbf{p})$}
\rput[r](-.1,1){1}
\rput[r](-.1,1.75){$\rate_{\text{seq}} $}
\rput[r](-.1,2.5){$\rate_{K-\text{mer}} $}
\end{pspicture}
\caption{\Fr obtained by the sequential algorithm is in the middle, given by $\rate_{\text{seq}} = {\bar LH_2(\p) \ln 2 \over \bar L H_2(\p) - 1}$ , \fr obtained by the $K$-mers based algorithm is at top, given by $\rate_{K-\text{mer}} = {\bar LH_2(\p)  \over \bar L H_2(\p) - 2}$.}
\label{fig:plotSeqKmers}
\end{figure}

%
%
%
%

The result is plotted in Fig. \ref{fig:plotSeqKmers}. The performance is strictly worse than that of the greedy algorithm. We give only a heuristic argument for Proposition~\ref{t:sequentialRate}.

 Motivated by the discussion in the previous section, we seek the typical overlap $\ell$ at which the first error occurs in merging a read; unlike the greedy algorithm, where this overlap corresponds to a specific stage of the algorithm, for the sequential algorithm this error can occur anytime between the first and last merging.

Let us compute the expected number of pairs of reads which can be merged in error at overlap $\ell$. To begin, a read has the potential to be merged to an incorrect successor at overlap $\ell$ if it has overlap less than or equal to $\ell$ with its true successor, since otherwise the sequential algorithm discovers the read's true successor. By the Poisson approximation, there are roughly $Ne^{-\lam(L-\ell)}$ reads with physical overlap less than or equal to $\ell$ with their successors. In particular, if $\ell < L-\lam\inv\ln N$ there will be no such reads, and so we may assume that $\ell$ lies between $L-\lam\inv \ln N$ and $L$.

Note furthermore that in order for an error to occur, the second read must not yet have been merged when the algorithm encounters the first read, and thus the second read must be positioned later in the sequence. This adds a factor one-half. Combining this reasoning with the preceding paragraph, we see that there are approximately $$\sfrac12N^2e^{-\lam(L-\ell)}$$ pairs of reads which may potentially be merged incorrectly at overlap $\ell$.

For such a pair, an erroneous merging actually occurs if the length-$\ell$ suffix of the first read equals the length-$\ell$ prefix of the second. Assuming (as in the greedy algorithm calculation) that these substrings are physically disjoint, the probability of this event is $2^{-\ell H_2(\p)}$.

The expected number of pairs of reads which are merged in error at overlap $\ell$, for $L-\lam\inv \ln N\leq \ell \leq L$, is thus approximately
\begin{equation}
\label{e:seqExp}
N^2e^{-\lam(L-\ell)}2^{-\ell H_2(\p)}\,.
\end{equation}
This number is largest when $\ell=L$ or $\ell = L-\lam\inv \ln N$, so the expression in \eqref{e:seqExp}
approaches zero if and only if ${N \over \Ncov} > {\bar LH_2(\p) \ln 2 \over \bar L H_2(\p) - 1}$  and $\bar{L}>2/H_2(\p)$, as in Proposition~\ref{t:sequentialRate}.

\subsubsection{$K$-mer based Algorithms}

Due to complexity considerations, many  recent assembly algorithms operate on $K$-mers instead of directly on the reads themselves. $K$-mers are length $K$ subsequences of the reads; from each read, one can generate $L-K+1$ $K$-mers. One of the early works which pioneer this approach is the sort-and-extend technique in ARACHNE \cite{Bat02}. By lexicographically sorting the set of all the $K$-mers generated from the collection of reads, identical $K$-mers from physically overlapping reads will be adjacent to each other. This enables the overlap relation between the reads (so called overlap graph)  to be computed in $O(N \log N)$ time (time to sort the set of $K$-mers) as opposed to the $O(N^2)$ time needed if pairwise comparisons between the reads were done. Another related approach is the De Brujin graph approach \cite{IW95,PTW01}. In this approach, the $K$-mers are represented as vertices of a De Brujin graph and there is an edge between two vertices if they represent adjacent $K$-mers in some read (here adjacency means their positions are offset by one). The DNA sequence reconstruction problem is then formulated as computing an Eulerian cycle traversing all the edges of the De Brujin graph.

The performance of these algorithms on the basic statistical model can be analyzed by observing that two conditions must be satisfied for them to work.

First, $K$ should be chosen such that with high probability, $K$-mers from physically disjoint parts of the DNA sequence should be distinct, i.e. there are no repeats of length $K$ subsequences in the DNA sequence. In the sort-and-extend technique, this will ensure that two identical adjacent $K$-mers in the sorted list belong to two physically overlapping reads rather than two physically disjoint reads. In the De Brujin graph approach, this will ensure that the Eulerian path will be connecting $K$-mers that are physically overlapping. This minimum $K$ can be calculated as we did to justify Lemma \ref{t:repeat}:
\begin{equation}
\label{eq:min_K}
\frac{K}{\log G} > \frac{2}{H_2({\bf p})}\,.
\end{equation}

Second, all successive reads should have physical overlap of at least $K$ base pairs. This is needed so that the reads can be assembled via the $K$-mers. According to the Poisson approximation, the expected number of successive reads with spacing greater than $L-K$ base pairs is roughly $N e^{-\lambda (L-K)}$. To ensure that with high probability all successive reads have overlap at least $K$ base pairs, this expected number should be small, i.e.
\begin{equation}
N > \frac{G\ln N }{L-K} \approx \frac{G \ln G }{L-K}\,.
\end{equation}
Substituting Eq. (\ref{eq:min_K}) into this and using the definition $\bar{L} = L/\log G$, we obtain
$$ \frac{N}{\Ncov} >  {\bar L H_2(\p) \over \bar L H_2(\p) -2}\,.$$

\Fr of this algorithm is plotted in Figure \ref{fig:plotSeqKmers}. Note that the  performance the $K$-mer based algorithms is strictly less than the performance achieved by the greedy algorithm. The reason is that for $\bar{L} > 2/H_2({\bf p})$, while the greedy algorithm only requires the reads to cover the DNA sequence, the $K$-mer based algorithms need more, that successive reads have (normalized) overlap at least $2/H_2({\bf p})$.

\subsection{Complexity of the Greedy Algorithm}

A naive implementation of the greedy algorithm would require an all-to-all pairwise comparison between all the reads. This would require a complexity of $O(N^2)$ comparisons. For $N$ in the order of tens of millions, this is not acceptable.  However, drawing inspiration from the sort-and-extend technique discussed in the previous section, a more clever implementation would yield a complexity of $O(LN\log N)$. Since $L \ll N$, this is a much more efficient implementation. Recall that in stage $\ell$ of the greedy algorithm, successive reads with overlap $\ell$ are considered. Instead of doing many pairwise comparisons to obtain such reads, one can simply extract all the $\ell$-mers from the reads and perform a sort-and-extend to find all the reads with overlap $\ell$. Since we have to apply sort-and-extend in each stage of the algorithm, the total complexity is $O(LN \log N)$.

An idea similar to this and resulting in the same complexity was described by Turner \cite{Tur89} (in the context of the shortest common superstring problem), with the sorting effectively replaced with a suffix tree data structure. Ukkonen \cite{Ukk90} used a more sophisticated data structure, which essentially computes overlaps between strings in parallel, to reduce the complexity to $O(NL)$.

\section{Markov DNA Model}

\label{sec:source}

In this section we extend the results for the basic i.i.d. DNA sequence model to a Markov sequence model.

\subsection{Formulation and Result}
The problem formulation is identical to the one in Section \ref{sec:formulation} except that we assume the DNA sequence is correlated and model it by a Markov source with transition matrix $Q=[q_{ij}]_{i,j\in \{1,2,3,4\}}$, where $q_{ij}=\P(S_k = i | S_{k-1}=j)$.

\begin{remark}
We assume that the DNA is a Markov process of order 1, but the result can be generalized to  Markov processes of order $m$ as long as $m$ is constant and does not grow with $G$.
\end{remark}

In the basic i.i.d. model, we observed that \fr depends on the DNA statistics through the R\'enyi entropy of order~2. We prove that a similar dependency holds for Markov models. In \cite{RAL01}, it is shown that the R\'enyi entropy {\em rate} of order~2 for a stationary ergodic Markov source with transition matrix $Q$ is given by
$$ H_2(Q) := \log \left({1\over \rho_{\max}(\bar Q)}\right),$$
where $\rho_{\max} (\bar Q) \triangleq \max\{ | \rho | : \rho ~\text{eigenvalue of}~ \bar Q\},$ and  $\bar Q = [q_{ij}^2]_{i,j\in \{1,2,3,4\}}$. In terms of this quantity, we state the following theorem.

\begin{theorem}\label{t:markovCapacity}
\Fr of a stationary ergodic Markov DNA sequence  is given by
\begin{equation}
\label{eq:markovCapacity}
\capacity (\bar L)=
\begin{cases}
\infty & \text{if} ~ \bar{L}  < 2/ H_2(Q) ,\\
1 & \text{if} ~ \bar{L} > 2/ H_2(Q).
\end{cases}
\end{equation}
\end{theorem}

\subsection{Sketch of Proof}
Similar to the i.i.d. case, it suffices to show the following statements:
\begin{enumerate}
\item If $\bar L < {2 \over H_2(Q)}$, $\Nmin(\epsilon,G,L) = \infty$ for sufficiently large $G$.
\item If $\bar L > {2 \over H_2(Q)}$, then  $\capacity =1 $.
\end{enumerate}

The following lemma is the analogue of Lemma \ref{t:repeat} for the Markov case, and is used similarly to prove statement $1$.

\begin{lemma}\label{l:markovConverse}
If $\bar{L}  < 2/H_2(Q) $, then a Markov DNA sequence contains interleaved repeats with probability $1-o(1)$.
\end{lemma}

To justify Lemma \ref{l:markovConverse} we use a similar heuristic argument as for the i.i.d. model, but with a new value for the probability that two physically disjoint  sequences $\s_i^L$ and $\s_j^L$ are equal:
$$\P(\s_i^L=\s_j^L)\approx e^{-L  \log \left({\rho_{\max}(\bar Q)}\right) }.$$
The lemma follows from the fact that  there are roughly $G^2$ such pairs in the DNA sequence. A formal proof of the lemma is provided in Appendix \ref{s:proofMarkovConverse}.

Statement 2 is again a consequence of the optimality of the greedy algorithm, as shown in the following lemma.

\begin{lemma}\label{p:MarkvoAchievability}
The greedy algorithm with exactly the same overlap score as used for the i.i.d. model can achieve minimum normalized coverage depth $\capacity = 1$  if $\bar L > {2 \over H_2(Q)}$.
\end{lemma}

Lemma \ref{p:MarkvoAchievability} is proved in Appendix \ref{s:markovAchievabilityProof}. The key technical contribution of this result is to show that the effect of physically overlapping reads does not affect the asymptotic performance of the algorithm, just as in the i.i.d. case.

\section{Noisy Reads}
\label{sec:noise}
In our basic model, we assumed that the read process is noiseless. In this section, we assess the effect of noise on the greedy algorithm.

\subsection{Formulation and Result}
\label{sec:formulation-noisy}
The problem formulation here differs from Section \ref{sec:formulation} in two aspects. First,  we assume that the read process is noisy and consider a simple probabilistic model for the noise. A nucleotide $s$ is read to be $r$ with probability $Q(r|s)$. Each nucleotide is perturbed independently, i.e. if $\br$ is a read from the physical underlying subsequence $\bs$ of the DNA sequence, then
$$\P(\br | \bs) =\prod_{i=1}^L Q(r_i| s_i).$$
Moreover, it is assumed that the noise affecting different reads is independent.

Second, we require a weaker notion of reconstruction. Instead of {\em perfect reconstruction}, we aim for {\em perfect layout}. By perfect layout, we mean that all the reads are mapped correctly to their true locations. Note that perfect layout does not imply perfect reconstruction as the consensus sequence may not be identical to the DNA sequence. On the other hand, since coverage implies that most positions on the DNA are covered by many reads ($O(\log G)$, to be more precise), the consensus sequence will be correct in most positions if we achieve perfect layout.

\begin{remark}
In the jargon of information theory, we are modeling the noise in the read process as a {\em discrete memoryless channel} with transition probability $Q(\cdot|\cdot)$. Noise processes in actual sequencing technologies can be more complex than this model. For example, the amount of noise can increase as the read process proceeds, or there may be insertions and deletions in addition to substitutions. Nevertheless, understanding the effect of noise on the assembly problem in this model provides considerable insight to the problem.
\end{remark}

We now evaluate the performance of the greedy algorithm for the noisy read problem. Finding the optimal algorithm for this case is an open problem.

To tailor the greedy algorithm for the noisy reads, the only requirement is to define the overlap score between two distinct reads. Given two reads $\br_i$ and $\br_j$, we would like to know whether they are physically overlapping with length $\ell$. Let $\bX$ and $\bY$ of length $\ell$ be the suffix of $\br_i$ and prefix of $\br_j$, respectively. We have the following hypotheses for $\bX$ and $\bY$:
\begin{itemize}
\item $H_0$: $\bX$ and $\bY$ are noisy reads from the same physical source subsequence;
\item $H_1$: $\bX$ and $\bY$ are noisy reads from two disjoint source subsequences.
\end{itemize}
The decision rule that is optimal in trading off the two types of error is the {\em maximum a posteriori} (MAP) rule, obtained by a standard large deviations calculation (see for example Chapter 11.7 and 11.9 of \cite{CT06}.)
In log likelihood form, the  MAP rule for this hypothesis testing problem is:
\begin{equation}\label{eq:criterion for decision}
\mbox{Decide $H_0$ if $\log \frac{P(\bx,\by)}{P(\bx)P(\by)} =\sum_{j=1}^{\ell} \log \frac{P_{X,Y}(x_j,y_j)}{P_X(x_j)P_Y(y_j)} \geq \ell \theta$},
\end{equation}
where $P_{X,Y}(x,y), P_X(x)$ and $P_Y(y)$ are the marginals of the joint distribution $P_{S}(s) Q(x | s) Q(y |s)$, and $\theta$ is a parameter reflecting the prior distribution of $H_0$ and $H_1$.

We can now define the overlap score, whereby two reads $\bR_i$ and $\bR_j$ have overlap at least $\ell$ if the MAP rule on the length $\ell$ suffix of $\bR_i$ and the length $\ell$ prefix of read $\bR_j$ decides $H_0$. The performance of the greedy algorithm using this score is given in the following theorem.

\begin{theorem}\label{p:general_IIDachievability}
The modified greedy algorithm can achieve normalized coverage depth $c(\bar L) =1$ if $\bar{L}>2/I^*$, where
$$I^*= \max_{\theta} \min (2D(P_{\mu} ||P_{X,Y}),D(P_{\mu} ||P_X \cdot P_Y)),,$$ and the distribution $P_\mu$ is given by
$$P_\mu(x,y) := \frac{[P_{X,Y}(x,y)]^{\mu}[P_X(x)P_Y(y)]^{1-\mu}}{\sum_{a,b} [P_{X,Y}(a,b)]^{\mu}[P_X(a)P_Y(b)]^{1-\mu}}$$ with $\mu$ the solution to the equation
$$ D(P_{\mu} ||P_X \cdot P_Y)-D(P_{\mu} ||P_{X,Y})=\theta.$$

\end{theorem}

The statement of Theorem~\ref{p:general_IIDachievability} uses the KL Divergence $D(P||Q)$ of the distribution $P$ relative to $Q$, defined as
\begin{equation}
D(P|| Q) = \sum_{a} P(a) \log {\frac{P(a)}{Q(a)}}.
\end{equation}
The details of the proof of the theorem are in Appendix \ref{sec:noisy-reads}. To illustrate the main ideas, we sketch the proof for the special case of uniform source and symmetric noise.

\subsection{Sketch of Proof for Uniform Source and Symmetric Noise}
In this section, we provide an argument to justify Theorem~\ref{p:general_IIDachievability} in the case of uniform source and symmetric noise. Concretely, $\bp = (1/4,1/4,1/4,1/4)$ and the noise is symmetric with transition probabilities:
\begin{equation}
Q(i | j)=
\begin{cases}
1-\epsilon ~& \text{if} ~ i=j \\
\epsilon/3 ~ & \text{if} ~ i\neq j\, .
\end{cases}
\end{equation}
The parameter $\epsilon$ is often called the error rate of the read process. It ranges from $1\%$ to $10$\% depending on the sequencing technology.

\begin{corollary}\label{p:Noisyachievability}
The greedy algorithm with the modified definition of overlap score between reads can achieve normalized coverage depth  $c(\bar{L})=1$ if $\bar{L}>2/I^*(\epsilon)$, where
$$ I^*(\epsilon) = D(\alpha^*||\tfrac{3}{4})$$
and $\alpha^*$ satisfies
$$ D(\alpha^*||\tfrac{3}{4}) = 2 D(\alpha^*||2\epsilon  - \tfrac{4}{3} \epsilon^2).$$
Here, $D(\alpha||\beta)$ is the divergence between a ${\rm Bern}(\alpha)$ and a ${\rm Bern}(\beta)$ random variable.
\end{corollary}
\begin{proof}
The proof follows by applying Theorem \ref{p:general_IIDachievability}. For uniform source and symmetric noise, the optimum $I^*$ is attained when $2D(P_{\mu} ||P_{X,Y})=D(P_{\mu} ||P_X \cdot P_Y)$.  The result is written in terms of $\alpha^*$ which is a function of the optimal value $\theta^*$.
\end{proof}

%

The performance of this algorithm is shown in Figure \ref{fig:plot_noisy_assembly}. The only difference between the two is a larger threshold, $2/I^*(\epsilon)$ at which \fr becomes one. A plot of this threshold as a function of $\epsilon$ is shown in Figure \ref{fig:I*}. It can be seen that when $\epsilon = 0$, $2/I^*(\epsilon)  = 2/H_2(\bp) = 1$, and increases continuously as $\epsilon$ increases.


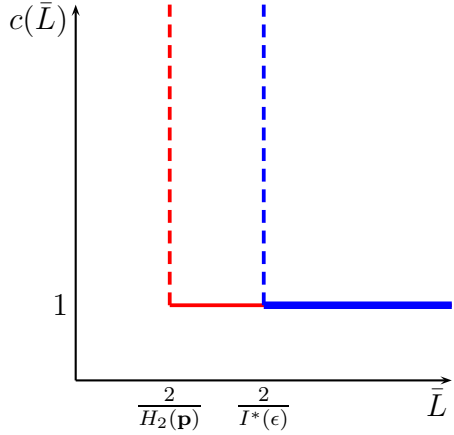
\begin{figure}\centering
\begin{pspicture}(-1,-1)(5,5)
\psline{->}(0,0)(5,0)
\psline{->}(0,0)(0,5)
\rput[t](4.8,-.1){$\bar L $}
\psline[linecolor=red,linestyle=dashed,linewidth=.05](1.25,1)(1.25,5)
\psline[linecolor=red,linewidth=.05](1.25,1)(5,1)
\rput[r](-.1,4.8){$c(\bar{L})$}
\rput[t](1.25,-.1){$2\over H_2(\mathbf{p})$}
\rput[r](-.1,1){1}

\psline[linecolor=blue,linestyle=dashed,linewidth=.05](2.5,1)(2.5,5)
\psline[linecolor=blue,linewidth=.1](2.5,1)(5,1)
\rput[t](2.5,-.1){$2\over I^*(\epsilon)$}
\end{pspicture}
\caption{The performance of the modified greedy algorithm with noisy reads.}
\label{fig:plot_noisy_assembly}
\end{figure}

%
%
%
%

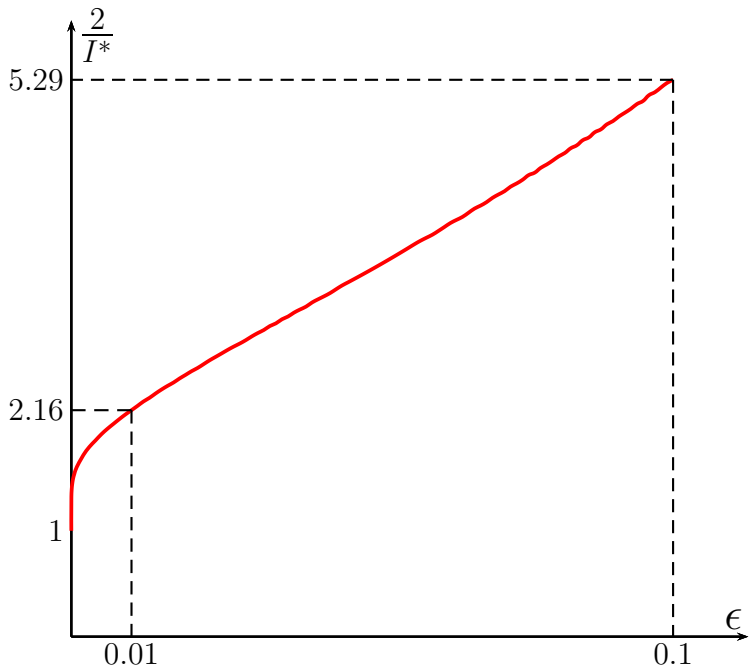
\begin{figure}
\centering

\readdata{\awgn}{noisy_read.data}

\centering

\begin{pspicture}(-.1,-.1)(9,8)
\psaxes[ticks=none,labels=none]{->}(0,0)(9,8.2)
\dataplot[plotstyle=curve,linecolor=red,linewidth=.05]{\awgn}

\psline[linestyle=dashed](8 , 7.4091)(8 ,0)
\rput[t]( 8,-.1){0.1}
\psline[linestyle=dashed](8 ,  7.4091)(0 , 7.4091)
\rput[r](-.1, 7.4091){5.29}

\psline[linestyle=dashed](0.8  ,  3.0133)(0.8  ,  0)
\rput[t](0.8,-.1){0.01}
\psline[linestyle=dashed](0.8,  3.0133)(0   , 3.0133)
\rput[r](-.1,3.0133){2.16}

\rput[r](-.1,1.4091){1}

\rput[l](.1,8){\Large $\frac{2}{I^*}$}
\rput[b](8.8,.1){\Large $\epsilon $}

\end{pspicture}

\caption{Plot of $2\over I^*(\epsilon)$  as a function of ~$\epsilon$ for the uniform source and  symmetric noise model.}
\label{fig:I*}
\end{figure}

We  justify the corollary by the following argument. In the noiseless case, two reads overlap by at least $\ell$ if the length $\ell$ prefix of one read is identical to the length $\ell$ suffix of the other read. The overlap score is the largest such $\ell$. When there is noise, this criterion is not appropriate. Instead, a natural modification of this definition is that two reads overlap by at least $\ell$ if the Hamming distance between the  prefix and the suffix strings is less than a fraction $\alpha$ of the length $\ell$. The overlap score between the two reads is the largest such $\ell$. The parameter $\alpha$ controls how stringent the overlap criterion is. By optimizing over the value of $\alpha$, we can obtain the following result.

We picture the greedy algorithm as working in stages, starting with an overlap score of $L$ down to an overlap score of 0. Since the spacing between reads is independent of the DNA sequence and noise process, the number of reads at stage $\ell$ given no errors have occurred in previous stages is again roughly
$$Ne^{-\lambda (L -\ell)}.$$
To pass this stage without making an error, the greedy algorithm should correctly merge those reads having spacing of length $\ell$ to their successors. Similar to the noiseless case, the greedy algorithm makes an error if the overlap score between two non-consecutive reads is $\ell$ at stage $\ell$, in other words
\begin{enumerate}
\item[1.] The Hamming distance between the length $\ell$ suffix of the present read and the length $\ell$ prefix of some read which is not the successor is less than $\alpha \ell$ by random chance.
\end{enumerate}
A standard large deviations calculation shows that the probability of this event is approximately
$$ 2^{-\ell D(\alpha||{3\over 4})},$$
which is the probability that two independent strings of length $\ell$ have Hamming distance less than $\alpha \ell$. Hence, the expected number of pairs of contigs for which this confusion event happens is approximately
\begin{equation}\label{eq:exp_f}
\left[ Ne^{-\lambda (L -\ell)} \right]^2 2^{-\ell D(\alpha||{3\over 4})}.
\end{equation}

Unlike the noiseless case, however, there is another important event affecting the performance of the algorithm. The \emph{missed detection} event is defined as
\begin{enumerate}
\item[2.] The Hamming distance between the length $\ell$ suffix of the present read and the length $\ell$ prefix of the successor read is larger than $\alpha \ell$ due to an excessive amount of noise.
\end{enumerate}
Again, a standard large deviations calculation shows that the probability of this event for a given read is approximately
$$ 2^{-\ell D(\alpha||\eta)},$$
where $\eta = 2\epsilon  - \frac{4}{3} \epsilon^2$ is the probability that the $i$th symbol in the length $\ell$ suffix of the present read does not match the $i$th symbol in the length $\ell$ prefix of the successor read (here we are assuming that $\alpha > \eta$). Thus the expected number of contigs missing their successor contig at stage $\ell$ is approximately
\begin{equation}\label{eq:expmiss}
Ne^{-\lambda (L -\ell)} 2^{-\ell D(\alpha||\eta)}.
\end{equation}
Both Equations \eqref{eq:exp_f} and \eqref{eq:expmiss} are largest either when $\ell=L$ or $\ell=0$. Similarly to the noiseless case, errors do not occur at stage 0 if the DNA sequence is covered by the reads. The coverage condition guarantees no gap exists in the assembled sequence.  From~\eqref{eq:exp_f} and~\eqref{eq:expmiss} we see that no errors occur at stage $L$ if
$$\bar L = \frac{L}{\log G} > \frac{2}{\min( 2D(\alpha||\eta), D(\alpha||\frac{3}{4}))}.$$
Selecting $\alpha$ to minimize the right hand side results in the two quantities within the minimum being equal, which gives the result.

\section{Discussions and Future Work}
\label{sec:discussions}

This paper seeks to understand the basic data requirements for shotgun sequencing, and we obtain results for simple models. The models for the DNA sequence and read process in this paper serve as a starting point from which to pursue extensions to more realistic models. We discuss a few of the many possible extensions.

\paragraph{Long repeats. } Long repeats occur in many genomes, from bacteria to human. The repetitive nature of real genomes is understood to be a major bottleneck for sequence assembly. Thus a caveat of this paper is that
the DNA sequence models we have considered, both i.i.d. and Markov, exhibit only short-range correlations, and therefore fail to capture the long-range correlation present in complex genomes. Motivated by this issue, a follow-up work \cite{BBT13} extends the approach of this paper to \emph{arbitrary} repeat statistics, in particular the statistics of actual genomes. The read model considered in \cite{BBT13} is the same uniform noiseless model we consider.

We briefly summarize the results and approach of \cite{BBT13}. 
First, Ukkonen's condition that there be no interleaved or triple repeats of length at least $L-1$ is generalized to give a lower bound on the read length and the coverage depth required for reconstruction in terms of repeat statistics of the genome. Next, they design a de Brujin graph based assembly algorithm that can achieve very close to the lower bound for repeat statistics of a wide range of sequenced genomes. The approach results in a pipeline, which takes
as input a genome sequence and desired success
probability $1-\eps$, computes a few simple repeat statistics, and from these statistics computes a
feasibility plot that indicates for which $L$ and $N$ reconstruction is possible. 

\paragraph{Double-strandedness. } The DNA sequence is double stranded and consists of a sequence $\bs$ and its reverse complement $\tilde \bs$. Reads are obtained from either of the two strands, and a natural concern is whether this affects the results. It turns out that for the i.i.d. sequence model considered in this paper (as well as the Markov model), the asymptotic minimum normalized coverage depth remains the same. The optimal greedy algorithm is modified slightly by including the reverse complements of the reads as well as the originals, and stopping when there are two reconstructed sequences $\bs$ and $\tilde \bs$.  The heuristic argument follows by observing that the probability of error at stage $\ell$ given in \eqref{eq:pairs} is changed only be a factor two, which does not change the asymptotic result. The rigorous proof involves showing that the contribution from overlapping reads is negligible, where the notion of reads overlapping accounts for both the sequence and its reverse complement. 

\paragraph{Read process. } There are a number of important properties of the read process which can be incorporated into more accurate models. Beyond the substitution noise considered in this paper, some sequencing technologies (such as PacBio) produce insertions and deletions. Often bases come with quality scores, and these scores can be used to mitigate the effect of noise. Other interesting aspects include correlation in the noise from one base to another (e.g. typically producing several errors in a row), non-uniformity of the error rate within a read, and correlation of the noise process with the read content. Aside from noise, a serious practical difficulty arises due to the positions of reads produced by some sequencing platforms being biased by the sequence, e.g. by the GC content. 
Noise and sampling bias in the reads make assembly more difficult, but another important direction is to incorporate mate-pairs into the read model. Mate-pairs (or paired-end reads) consisting of two reads with an approximately known separation, help to resolve long repeats using short reads. 

\paragraph{Partial reconstruction. } In practice the necessary conditions for perfect reconstruct are not always satisfied, but it is still desirable to produce the best possible assembly. While the notion of perfect reconstruction is relatively simple, defining what ``best" means is more delicate for partial reconstructions; one must allow for multiple contigs in the output as well as errors (misjoins). Thus an optimal algorithm is one which trades off optimally between the number of contigs and number of errors.

\section*{Acknowledgements}

This work is supported by the Natural Sciences and Engineering
Research Council (NSERC) of Canada  and by the Center for Science of Information (CSoI), an NSF Science and Technology Center, under grant agreement CCF-0939370. The authors would like to thank Professor Yun Song for discussions in the early stage of this project, and Ramtin Pedarsani for discussions about the complexity of the greedy algorithm.

\appendix

\section{Proof of Theorem \ref{t:mainCapacity}, Part $2$}

We first state and prove the following lemma. This result can be found in \cite{Arr96}, but for ease of generalization to the Markov case later, we will include the proof.

\begin{lemma}\label{lem self match iid}
For any distinct substrings $\bX$ and $\bY$ of length $\ell$ of the i.i.d. DNA sequence:
\begin{enumerate}
\item If the strings have no physical overlap, the probability that they are identical is  $e^{-\ell H_2(\bp)}$.
\item If the strings have physical overlap, the probability that they are identical is  bounded above by $e^{-\ell H_2(\bp)/2}$.
\end{enumerate}
\end{lemma}

\begin{proof}
We first note that for any $k$ distinct bases in the DNA sequence the probability that they are identical is given by
$$\pi_k\triangleq \sum_{i=1}^{4} p_i^k.$$

1- Consider $\bX=S_{i+1}\dots S_{i+\ell}$ and $\bY=S_{j+1}\dots S_{j+\ell}$ have no physical  overlap. In this case, the events $\{S_{i+m}=S_{j+m}\}$ for $m\in\{1,\dots,\ell\}$ are independents and equiprobable. Therefore, the probability that  $\bX=\bY$ is given by
$$ \pi_2 ^ \ell = 2^{-\ell H_2(\bp)}.$$

 2- For the case of overlapping strings $\bX$ and $\bY$, we assume that a substring of length $k< \ell$ from the DNA sequence is shared between the two strings. Without loss of generality, we also assume that $\bX$ and $\bY$ are, respectively, the prefix and suffix of $\bS_1^{2\ell-k}$. Let $q$ and $r$ be the quotient and remainder of $\ell$ divided by $\ell-k$, i.e., $\ell=q(\ell-k)+r$ where $0\leq r< \ell-k$.
 It can be shown that  $\bX=\bY$ if and only if $\bS_1^{2\ell-k}$ is a string of the form $\bf U V U V \dots U V U $ where $\bf U$ and $\bf V$ have length $r$ and $\ell -k -r$. Since the number of $\bf U$ and $\bf V$ are, respectively, $q+2$ and $q+1$, the probability of observing a structure of the form $\bf U V U V \dots U V U $ is given by
\begin{align*}
(\pi_{q+2})^{r}\times (\pi_{q+1})^{\ell-k-r}  \stackrel{(a)}{\leq} (\pi_{2})^{\frac{r(q+2)}{2}}\times (\pi_{2})^{\frac{(\ell-k-r)(q+1)}{2}}=(\pi_2)^{\ell-\frac{k}{2}},
\end{align*}
where $(a)$ comes from the fact that $(\pi_{q})^{\frac{1}{q}} \leq (\pi_{2})^{1\over 2}$ for all $q\geq 2$.
Since $k<\ell$, we have $(\pi_2)^{\ell-\frac{k}{2}}\leq (\pi_2)^{\frac{\ell}{2}}$. Therefore, the probability that $\bX=\bY$ for two overlapping strings is bounded above by
$$(\pi_2)^{\frac{\ell}{2}} =2^{-\ell H_2(\bp)/ 2}.$$
This completes the proof.
\end{proof}

\noindent
{\bf Proof of Theorem \ref{t:mainCapacity}, Part $2$}

The greedy algorithm finds a contig corresponding to a substring of the DNA sequence if each read $\bR_i$ is correctly merged to its successor read $\bR_i^s$ with the correct amount of \emph{physical overlap} between them which is  $\ph_i=L-(T^s_i - T_i)_{ \mod G}$.\footnote{Note that the physical overlap can take negative values.} If, in addition, the whole sequence is covered by the reads then the output of the algorithm is exactly the DNA sequence $\s$.

Let $\EE_1$ be the event that some read is merged incorrectly; this includes merging to the read's valid successor but at the wrong relative position as well as merging to an impostor\footnote{A read $R_j$ is an impostor to $R_i$ if $ W(\bR_i,\bR_j)\geq V_i$.}. Let $\EE_2$ be the event that the DNA sequence is not covered by the reads. The union of these events, $\EE_1\cup \EE_2$, contains the error event $\EE$. We first focus on event $\EE_1$.

Since the greedy algorithm merges reads according to overlap score, we may think of the algorithm as working in stages starting with an overlap score of $L$ down to an overlap  score of $0$. Thus $\EE_1$ naturally decomposes as $\EE_1=\cup_\ell \calA_\ell$, where $\calA_\ell$ is the event that the first error in merging occurs at stage $\ell$.

Now, we claim that
\begin{equation}\label{e:Adecomb}
\calA_\ell \subseteq \BB_{\ell}\cup \CC_{\ell}\,,
\end{equation}
where:
\begin{eqnarray}
\BB_\ell & \triangleq & \{\bR_j\neq \bR_i^s, U_j\leq \ell, V_i\leq \ell, W_{ij}=\ell \mbox{ for some $i \not = j$.}\} \label{eq:eventB}\\
\CC_\ell & \triangleq & \{\bR_j=\bR_i^s, U_j=V_i< \ell, W_{ij}=\ell \mbox{ for some $i \not = j$.}\} \label{eq:eventC}
\end{eqnarray}

If the event $\calA_{\ell}$ occurs,  then either there are two reads $\bR_i$ and $\bR_j$'s such that $\bR_i$ is merged to its successor $\bR_j$ but at an overlap larger than their physical overlap, or  there are two reads $\bR_i$ and $\bR_j$ such that $\bR_i$ is merged to $\bR_j$, an impostor. The first case implies the event $\CC_\ell$. In the second case, in addition to $W_{ij} = \ell$, it must be true that the physical overlaps $V_i, U_j \le \ell$, since otherwise at least one of these two reads would have been merged at an earlier stage. (By definition of $\calA_{\ell}$, there were no errors before stage $\ell$). Hence, in this second case, the event $\BB_\ell$ occurs.

Now we will bound  $\P(\BB_\ell)$ and $\P(\CC_\ell)$.

First, let us consider the event $\BB_\ell$. This is the event that two reads which are not neighbors with each other got merged by mistake. Intuitively, event $\BB_\ell$ says that the  pairs of reads that can potentially cause such confusion at stage $\ell$ are limited to those with short physical overlap with their own neighboring reads, since the ones with large physical overlaps have already been successfully merged to their correct neighbor by the algorithm in the early stages. In Figure \ref{fig:contigs}, these are the reads at the ends of the contigs that are formed by stage $\ell$.

For any two distinct reads $\bR_i$ and $\bR_j$, we define the following event
$$\BB_\ell ^ {ij } \triangleq   \{\bR_j\neq \bR_i^s, U_j\leq \ell, V_i\leq \ell, W_{ij}=\ell \}$$
From the definition of $\BB_\ell$ in~\eqref{eq:eventB}, we have $\BB_\ell \subseteq \cup_{ij} \BB_\ell ^{ij}$. Applying the union bound and considering the fact that $\BB_\ell ^{ij}$'s are equiprobable yields
$$\P(\BB_\ell) \leq N^2 \P(\BB_\ell^{12}).$$

Let $\DD$ be the event that the two reads $\bR_1$ and $\bR_2$ have no physical overlap.
Using the law of total probability  we obtain
\begin{align*}
\P(\BB_\ell^{12}) & = \P(\BB_\ell^{12}| \DD)\P(\DD)+\P(\BB_\ell^{12}|\DD^c)\P(\DD^c).
\end{align*}
Since $\DD^c$ happens only if $T_2 \in [T_1-L+1,T_1+L-1]$, $\P(\DD^c) \le \frac{2L}{G}$.  Hence,
\begin{align}\label{e:DtotProb}
\P(\BB_\ell^{12})& \leq  \P(\BB_\ell^{12}| \DD)+\P(\BB_\ell^{12}|\DD^c)\frac{2L}{G}.
\end{align}
We proceed with bounding  $\P(\BB_\ell^{12}| \DD)$ as follows,
\begin{align*}
\P(\BB_\ell^{12}| \DD) & = \P(U_2\leq \ell, V_1\leq \ell, W_{12}=\ell |\DD )\\
&\stackrel{(a)}{=}  \P(U_2\leq \ell, V_1\leq \ell | \DD )\P(  W_{12}=\ell | \DD )\\
&\stackrel{(b)}{=}  \P(U_2\leq \ell, V_1\leq \ell |  \DD )e^{- \ell H_2(\bp)},
\end{align*}
where $(a)$ comes from the fact that given $\DD$, the events $\{U_2\leq \ell, V_1\leq \ell \}$ and $\{W_{12}=\ell\}$ are independent, and $(b)$ follows from Lemma  \ref{lem self match iid} part 1.

Note that the event $\{\bR_2\neq \bR_1^s, U_2\leq \ell, V_1\leq \ell \}$ implies that no reads start in the intervals $[T_1, T_1+L-\ell-1]$ and $[T_2-L+\ell+1,T_2]$. Given $\DD$, if the two intervals overlap then there exists a read with starting position $T_i$ with $T_i\in [T_1, T_1+L-\ell-1]$ or $T_i\in [T_2-L+\ell+1,T_2]$. To see this, suppose $T_i$ is not in one of the intervals then $\bR_2$ has to be the successor of $\bR_1$ contradicting $\bR_2\neq \bR_1^s$.  If the two intervals are disjoint, then the probability that there is no read starting in them is given by
$$ \left(1-\frac{2(L-\ell)}{G} \right)^{N-2}.$$
Using the inequality $1-a\leq e^{-a}$, we obtain
\begin{align}\label{e:B1stTermBound}
\P(\BB_\ell^{12}| \DD)\leq e^{-2\lambda (L-\ell)(1-2/N)} 2^{- \ell H_2(\bp)}
\end{align}

To bound  $\P(\BB_\ell^{12}| \DD^c)$, we note that it has a non-zero value only if the length of physical overlap between $\bR_1$ and $\bR_2$ is less than $\ell$. Hence, we only consider physical overlaps of length less that $\ell$ and denote this event by $ \DD_1$. We proceed as follows,
\begin{align*}
\P(\BB_\ell^{12}| \DD^c) & \leq \P(U_2\leq \ell, V_1\leq \ell, W_{12}=\ell |\DD_1 )\\
& \leq \P( V_1\leq \ell, W_{12}=\ell |\DD_1 )\\
&\stackrel{(a)}{\leq }  \P(V_1\leq \ell |\DD_1 )\P(  W_{12}=\ell |\DD_1 )\\
&\stackrel{(b)}{\leq }  \P(V_1\leq \ell |\DD_1 ) 2^{-\ell H_2(\bp)/2}
\end{align*}
where $(a)$ comes from the fact that given $\DD_1$, the events $\{ V_1\leq \ell \}$ and $\{W_{12}=\ell\}$ are independent, and $(b)$ follows from Lemma  \ref{lem self match iid} part 2.  Since $\{ V_1\leq \ell \}$ corresponds to the event that there is no read starting in the interval $[T_1, T_1+L-\ell-1]$, we obtain
$$\P(V_1\leq \ell |\bar \DD_2 ) =\left(1-\frac{L-\ell}{G} \right)^{N-2}.$$
Using the inequality $1-a\leq e^{-a}$, we obtain
\begin{align*}
\P(\BB_\ell^{12}| \DD_2) &\leq    e^{-\lambda (L-\ell)(1-2/N)} 2^{- \ell H_2(\bp)/2}.
\end{align*}

Putting all terms together, we have
\begin{align}
\label{eq:pb}
\P(\BB_\ell) &\leq  q_\ell ^2 + 2\lambda L q_\ell.
\end{align}
where
\begin{align}\label{eq:ql}
q_\ell  = \lambda G e^{-\lam(L-\ell)(1-2/N)} 2^{-\ell H_2(\p)/2}.
\end{align}
The first term reflects the contribution from the reads with no physical overlap and the second term from the reads with physical overlap. Even though there are lots more of the former than the latter, the probability of confusion when the reads are physically overlapping can be much larger. Hence both terms have to be considered.



Let us define
$$\CC_\ell ^i  \triangleq  \{V_i< \ell, W(\bR_i, \bR_i^s)=\ell \}. $$
From the definition of $\CC_\ell$ in~\eqref{eq:eventC}, we have $\CC_\ell \subseteq \cup_{i} \CC_\ell ^{i}$. Applying the union bound and considering the fact that $\CC_\ell ^{i}$'s are equiprobable yields
$$\CC_\ell \leq N \P(\CC_\ell ^1). $$
Hence,
\begin{align*}
\P(\CC_\ell) & \leq N \P( W(\bR_i, \bR_i^s)=\ell  | V_i< \ell) \P(V_i< \ell)
\end{align*}
Applying Lemma  \ref{lem self match iid} part 2, we obtain
\begin{align*}
\P(\CC_\ell) & \leq N e^{- \ell H_2(\bp)/2} \left(1-\frac{L-\ell}{G} \right)^{N-1}.
\end{align*}
Using the inequality $1-a\leq e^{-a}$, we obtain
\begin{align}
\nonumber
\P(\CC_\ell) & \leq \lambda G e^{-\lambda (L-\ell)(1-1/N)} 2^{- \ell H_2(\bp)/2} \\
& \leq q_\ell \label{eq:pc}
\end{align}



Using the bounds, \eqref{eq:pb} and \eqref{eq:pc}, we get
\begin{align*}
\P(\EE_1)  =  \P(\cup_\ell \calA_\ell)
\le  \sum_{\ell =0}^L \P(\calA_\ell)
=  \sum_{\ell = 0}^L \P(\BB_\ell) + \P(\CC_\ell)
\le   \sum_{\ell = 0}^L q_\ell ^2 + (2 \lambda L+1)  q_\ell,
\end{align*}
where $q_\ell$ is defined in \eqref{eq:ql}. Since $q_\ell$ is monotonic in $\ell$, we can further bound $\P(\EE_1)$ by:
\begin{equation}
\label{eq:bb}
 \P(\EE_1) \le (L+1) \max \left \{ q_0^2+(2\lambda L + 1) q_0,  q_L^2 + (2\lambda L + 1) q_L^2 \right \}.
 \end{equation}

Since $\bar{L}>\frac{2}{H_2(\p)}$,  $q_L$ vanishes exponentially in $L$ and the second term on the right hand side of (\ref{eq:bb}) has no contribution asymptotically. Now, choose
$$ N = \frac{G}{L} \ln (GL^3).$$
A direct computation shows that for this choice of $N$, $q_0^2+(2\lambda L + 1) q_0 = O (\frac{1}{L^2} ).$ Hence,  the bound (\ref{eq:bb}) implies that $\P(\EE_1) \rightarrow 0$. Moreover the probability of no coverage $\P(\EE_2)$ also goes to zero with this choice of $N$. Hence, the probability of error in reconstruction $\P(\EE)$ also goes to zero. This implies that the minimum number of reads required to meet the reconstruction error probability of at most $\eps$ satisfies:
$$ \Nmin(\eps,G,L) \le \frac{G}{L} \ln (GL^3)$$
for sufficiently large $G$ and $L$ with $L/\log G = \bar{L}$.  Hence, this implies that
$$ \limsup_{L,G \rightarrow \infty, L/\log G = \bar{L}} = \frac{\Nmin(\eps,G,L)}{G/\bar{L}} \le 1.$$
Combining this with Lemma \ref{lem:cov_asym}, we get:
$$ \limsup_{L,G \rightarrow \infty, L/\log G = \bar{L}} = \frac{\Nmin(\eps,G,L)}{\Ncov(\eps,G,L)} \le 1.$$
But since $\Nmin(\eps,G,L) \ge \Ncov(\eps,G,L)$, it follows that:
$$ \lim_{L,G \rightarrow \infty, L/\log G = \bar{L}} = \frac{\Nmin(\eps,G,L)}{\Ncov(\eps,G,L)} = 1,$$
completing the proof.

\section{Proof of Theorem \ref{t:markovCapacity}}
The stationary distribution of the source is denoted by $\bp=(p_1,p_2,p_3,p_4)^t$.
Since $\bar Q$ has positive entries, the Perron-Frobenius theorem implies that its largest eigenvalue $\rho_{\max}(\bar Q)$ is real and positive and the corresponding eigenvector $\boldsymbol{\pi}$ has positive components. The following inequality is useful:
\begin{align}
\nonumber \sum_{i_1 i_2 \dots i_\ell}  q_{i_2 i_1}^2 q_{i_3 i_2}^2\dots q_{i_{\ell} i_{\ell-1}}^2 & \leq \max_{i_1\in \{1,2,3,4\}}\left\{{1 \over \pi_{i_1}} \right\}  \sum_{i_1 i_2 \dots i_\ell} \pi_{i_1} q_{i_2 i_1}^2 q_{i_3 i_2}^2\dots q_{i_{\ell} i_{\ell-1}}^2 \\
\nonumber & = \max_{i\in \{1,2,3,4\}}\left\{{1 \over \pi_{i}} \right\} || \bar Q^{\ell-1} \boldsymbol{\pi} ||_1\\
\nonumber & =\max_{i\in \{1,2,3,4\}}\left\{{1 \over \pi_{i}} \right\}  \left(\rho_{\max}(\bar Q)\right)^{\ell-1} ||  \boldsymbol{\pi} ||_1 \\
\label{eq:markov-inequality} & =\gamma  \left(\rho_{\max}(\bar Q)\right)^{\ell}
\end{align}
where $\gamma=\max_{i\in \{1,2,3,4\}}\left\{{||  \boldsymbol{\pi} ||_1 \over \pi_{i}\rho_{\max}(\bar Q) } \right\} $.

\subsection{Proof of Lemma \ref{l:markovConverse} }\label{s:proofMarkovConverse}

In \cite{Arr96}, Arratia \emph{et al.} showed that interleaved pairs of repeats are the dominant term causing non-recoverability. They also used poisson approximation to derive bounds on the event that $\s$ is recoverable from its $L$-spectrum. We take a similar approach to obtain an upper bound under the Markov model. First, we state the following theorem regarding Poisson approximation of the sum of indicator random variables, c.f. Arratia \emph{et al.} \cite{Arratia90}.

\begin{theorem}[Chen-Stein Poisson approximation]\label{th:chen-stein}
Let $W=\sum_{\alpha\in I} \chi_{\alpha}$ where $\chi_{\alpha}$'s are indicator random variables for some index set $I$. For each $\alpha$, $B_\alpha\subseteq I $ denotes the set of indices where $\chi_\alpha$ is independent from the $\sigma$-algebra generated by  all $\chi_\beta$ with $\beta\in I-B_\alpha$.  Let
\begin{align}
\label{eq:chen1} b_1 & = \sum_{\alpha \in I} \sum_{\beta \in B_\alpha} \E[\chi_\alpha]\E[\chi_\beta],\\
\label{eq:chen2} b_2 & = \sum_{\alpha \in I} \sum_{\beta \in B_\alpha, \beta\neq \alpha} \E[\chi_\alpha \chi_\beta].
\end{align}
Then
\begin{equation}
d_{\text{TV}}(W,W')   \leq {1-e^{-\theta} \over \theta}(b_1+b_2),
\end{equation}
where $\theta=\E[W]$ and $d_{\text{TV}}(W,W')$ is the total variation distance\footnote{The total variation distance between two distributions $W$ and $W'$ is defined by $d_{\text{TV}}(W,W') =\sup_{A\in \mathcal{F}} |\P_W(A)-\P_{W'} (A) |  $, where $\mathcal{F}$ is the $\sigma$-algebra defined for $W$ and $W'$} between $W$ and Poisson random variable $W'$ with the same mean.

\end{theorem}

\paragraph{Proof of Lemma \ref{l:markovConverse}}

 Let $\UU$ denote the event that there is no two pairs of interleaved repeats in the DNA sequence. Given the presence of $k$ repeats in $\s$, the probability of $\UU$ can be found by using  the Catalan numbers \cite{Arr96}. This probability is $2^k /(k+1)!$. If $Z$ denotes the random variable indicating the number of repeats in the DNA sequence, we obtain,
$$\P(\UU)=\sum_{k} \frac{2^k}{(k+1)!}\P(Z=k).$$

To approximate $\P(\UU)$, we partition the sequence as
$$\s=S_1 \bX_1 S_{L+2} \bX_2 S_{2(L+1)+1}\bX_3 ~\dots S_{(K-1)(L+1)+1} \bX_K$$
where $\bX_i=\s [(i-1)(L+1)+2,i(L+1)]$ and $K={G\over L+1}$.
Each $\bX_{i}$ has length $L$ and will be denoted by $\bX_i=X_{i1}\dots X_{iL}$. We write $\bX_i \sim \bX_j$ with $i\neq j$ to mean  $X_{i1}\neq X_{j1}$ and $X_{ik}= X_{jk}$ for $2 \leq k \leq L$.  In other words, $\bX_i \sim \bX_j$ means that there is a repeat of length at least $L-1$ starting from locations $(i-1)(L+1)+3$ and $(j-1)(L+1)+3$ in the DNA sequence and the repeat cannot be extended from left. The requirement $X_{i1}\neq X_{j1}$ is due to the fact that allowing left extension ruins accuracy of Poisson approximation as repeats appear in clumps.

Let $I=\{(i,j) | 1\leq i < j\leq K \}$. Let $\chi_{\alpha}$ with $\alpha \in I$ denote the indicator random variable for a repeat at $\alpha=(i,j)$, i.e., $\chi_{\alpha}=\mathbf{1}(\bX_i \sim \bX_j)$.  Let $W=\sum_{\alpha \in I} X_{\alpha}$. Clearly,
$$\P(\UU) \leq \sum_{k} \frac{2^k}{(k+1)!}\P(W=k).$$

Letting  $\bY=S_1S_{L+2}\dots S_{(K-1)(L+1)+1}$, we obtain
$$\P(\UU) \leq \sum_{\bY}  \sum_{k}\frac{2^k}{(k+1)!}\P(W=k | \bY)\P(\bY).$$
For any $\bY$, let $\epsilon$ be the total variation distance between $W$ and its corresponding Poisson distribution $W'$ with mean $\theta_{\bY}=\E[W| \bY]$. Then, we obtain
\begin{align*}
\P(\UU) & \leq \sum_{\bY} \left( \epsilon+ e^{-\theta_\bY} \sum_{k=0}^{\infty} \frac{(2\theta_\bY)^k}{k!(k+1)!}\right)\P(\bY)\\
& = \epsilon +  \sum_{\bY} e^{-\theta_\bY} \sum_{k=0}^{\infty} \frac{(2\theta_\bY)^k}{k!(k+1)!} \P(\bY)\\
& \leq \epsilon +  \sum_{\bY} e^{-\theta_\bY} \sum_{k=0}^{\infty} \left( \frac{(\sqrt{2\theta_\bY})^k}{k!} \right)^2 \P(\bY)\\
& \leq \epsilon +  \sum_{\bY} e^{-\theta_\bY} \left( \sum_{k=0}^{\infty}  \frac{(\sqrt{2\theta_\bY})^k}{k!} \right)^2 \P(\bY)\\
& \leq \epsilon +  \sum_{\bY} e^{-\theta_\bY+2\sqrt{2\theta_\bY}} \P(\bY)
\end{align*}

We assume $\theta_\bY \geq 8$ for all $\bY$ and let $\theta=\min_{\bY} \theta_\bY$. For this region, the exponential factor within the summation is monotonically decreasing and
\begin{equation}\label{eq:interleave-bound}
\P(\UU) \leq \epsilon +  e^{-\theta+2\sqrt{2\theta}}.
\end{equation}

To calculate the bound, we need to obtain an upper bound for $\epsilon$ and a lower bound for $\theta$. We start with the lower bound on $\theta$.
From Markov property and for a given $\alpha=(i,j)$,
\begin{align*}
\E[\chi_\alpha | \bY ] & =  \sum_{i_1 i_2 \dots i_L} \P(X_{i1}\neq X_{j1} | \bY)  q_{i_2 i_1}^2 q_{i_3 i_2}^2\dots q_{i_{L} i_{L-1}}^2\\
& \geq \min_{} \left\{ \frac{\P(X_{i1}\neq X_{j1} | \bY)}{\pi_{i_1}}  \right\} \sum_{i_1 i_2 \dots i_L} \pi_{i_1}  q_{i_2 i_1}^2 q_{i_3 i_2}^2\dots q_{i_{L} i_{L-1}}^2\\
&= \zeta  \left(\rho_{\max}(\bar Q)\right)^{L}
\end{align*}
where $\zeta=\min_{} \left\{ \frac{\P(X_{i1}\neq X_{j1} | \bY)}{\pi_{i_1}\rho_{\max}(\bar Q)}  \right\} $. Therefore,

\begin{align}\label{eq:theta}
\theta_\bY =\sum_{\alpha\in I} \E[\chi_\alpha | Y ]\geq {K \choose 2} \zeta  \left(\rho_{\max}(\bar Q)\right)^{L}=\theta.
\end{align}

To bound $\epsilon$, we make use of the Chen-Stein method.  Let $B_{\alpha}=\{(i',j') \in I | i'=i~\text{or}~j'=j\}.$ Note that $B_{\alpha}$ has cardinality $2K-3$. Since given $\bY$, $\chi_{\alpha}$  is independent of the sigma-algebra generated by all $\chi_{\beta}$, $\beta\in I- B_{\alpha}$, we can use Theorem \ref{th:chen-stein} to obtain
\begin{equation}
 d_{\text{TV}}(W,W' | \bY) \leq \frac{b_1+b_2}{\theta_{\bY}},
\end{equation}
where $b_1$ and $b_2$ are defined in (\ref{eq:chen1}) and  (\ref{eq:chen2}), respectively. Since $\E[X_{\alpha}X_{\beta}] =\E [X_\alpha ]\E[X_\beta]$ for all $\alpha \neq \beta \in B_\alpha$, we can conclude that $b_2 \leq b_1$. Therefore,
\begin{equation}\nonumber
 d_{\text{TV}}(W,W' | \bY) \leq \frac{2b_1}{\theta_{\bY}}.
\end{equation}
Since $\theta \leq \theta_{\bY}$,
\begin{equation}\nonumber
 d_{\text{TV}}(W,W' | \bY) \leq \frac{2b_1}{\theta}.
\end{equation}
In order to compute $b_1$, we need an upper bound on $\E[\chi_\alpha | \bY ]$. By using (\ref{eq:markov-inequality}), we obtain
\begin{align*}
\E[\chi_\alpha | \bY ] & =  \sum_{i_1 i_2 \dots i_L} \P(X_{i1}\neq X_{j1} | \bY)  q_{i_2 i_1}^2 q_{i_3 i_2}^2\dots q_{i_{L} i_{L-1}}^2\\
& \leq \sum_{i_1 i_2 \dots i_L}   q_{i_2 i_1}^2 q_{i_3 i_2}^2\dots q_{i_{L} i_{L-1}}^2\\
& \leq \gamma  \left(\rho_{\max}(\bar Q)\right)^{L}.
\end{align*}
Hence,
\begin{align*}
b_1 & = \sum_{\alpha \in I} \sum_{\beta \in B_\alpha} \E[\chi_\alpha | \bY]\E[\chi_\beta | \bY],\\
&\leq (2K-3){K \choose 2} \gamma^2 \left(\rho_{\max}(\bar Q)\right)^{2L}\\
& = \frac{\gamma^2 \theta^2 (2K-3)}{\zeta^2 {K\choose 2}}\\
& \leq \frac{4\gamma^2 \theta^2 }{\zeta^2 K}.
\end{align*}
Using the bound for $b_1$, we have the following bound for the total variation distance.
\begin{equation}\nonumber
 d_{\text{TV}}(W,W' | \bY) \leq  \frac{8\gamma^2 \theta }{\zeta^2 K}.
\end{equation}
Form the above inequality, we can choose $\epsilon=\frac{8\gamma^2 \theta }{\zeta^2 K}$. Substituting in (\ref{eq:interleave-bound}) yields
\begin{equation}\label{eq:final-bound}
\P(\UU) \leq \frac{8\gamma^2 \theta }{\zeta^2 K} +  e^{-\theta+2\sqrt{2\theta}}.
\end{equation}

From the definition of $\theta$ in (\ref{eq:theta}), we have
$$\theta= {\zeta (K-1)(L+1)^2 \over 2K} G^{2-\bar L \log \left({1\over \rho_{\max}(\bar Q)}\right)}.$$
Therefore, if $2> \bar L \log \left({1\over \rho_{\max}(\bar Q)}\right)$ then $\theta$ and $\theta\over K$ go, respectively, to infinity and zero  exponentially fast.  Since the right hand side of (\ref{eq:final-bound}) approaches zero, we can conclude that with probability $1-o(1)$ there exists a two pairs of interleaved repeats in the sequence. This completes the proof.

\subsection{Proof of Lemma \ref{p:MarkvoAchievability}} \label{s:markovAchievabilityProof}

The proof follows closely from that of the i.i.d. model. In fact, we only need to replace Lemma \ref{lem self match iid} with the following lemma.

\begin{lemma}\label{lem self match markov}
For any distinct substrings $\bX$ and $\bY$ of length $\ell$ of the Markov DNA sequence:
\begin{enumerate}
\item If the strings have no physical overlap, the probability that they are identical is   bounded above by $\gamma  2^{\ell \log  \left(\rho_{\max}(\bar Q)\right)}$.
\item If the strings have physical overlap, the probability that they are identical is  bounded above by $ \sqrt{\gamma}  2^{\ell \log  \left(\rho_{\max}(\bar Q)\right)/2}$.
\end{enumerate}
\end{lemma}

%

\def\bS{\mathbf{S}}
\def\bU{\mathbf{U}}
\def\bV{\mathbf{V}}

\begin{proof}
1- From Markov property, we can show that
\begin{align*}
\P(\bX=\bY) & = \sum_{i_1 i_2 \dots i_\ell} \P(X_1=Y_1=i_1) q_{i_2 i_1}^2 q_{i_3 i_2}^2\dots q_{i_{\ell} i_{\ell-1}}^2\\
& \leq \sum_{i_1 i_2 \dots i_\ell} q_{i_2 i_1}^2 q_{i_3 i_2}^2\dots q_{i_{\ell} i_{\ell-1}}^2\\
&\leq  \gamma \left(\rho_{\max}(\bar Q)\right)^{\ell},
\end{align*}
where the last line follows from (\ref{eq:markov-inequality}).

2- Without loss of generality, we assume that $\bX=\bS [1, \ell]$ and $\bY=\bS[\ell-k+1,2\ell - k]$ for some $k\in\{1,\dots,\ell-1\}$. Let $q$ and $r$ be the quotient and remainder of dividing $2\ell-k$ by $\ell-k$. From decomposition of  $\bS [1, 2\ell - k]$  as $\bU_1 \bU_2 \dots \bU_q \bV$ where $|\bU_i|=\ell-k$ for all $i\in\{1,\dots,q\}$ and $|\bV|=r$, one can deduce that $\bX=\bY$ if and only if $\bU_i=S_1S_2\dots S_{\ell-k}$ for all $i\in\{1,\dots,q\}$ and $\bV= S_1 S_2 \dots S_r$. Hence, we have
\begin{align*}
\P(\bX=\bY) & = \P(\bS [1, 2\ell - k]=\bU \bU \dots \bU \bV)\\
& =\sum_{i_1 i_2 \dots i_{\ell-k}} p_{i_1}  \left( q_{i_2 i_1} q_{i_3 i_2}\dots q_{i_1 i_{\ell-k}}\right)^q \left( q_{i_2 i_1} q_{i_3 i_2}\dots q_{i_r i_{r-1}}\right) \\
& \stackrel{(a)}{\leq } \sqrt {\sum_{i_1}p_{i_1}^2} \sqrt{ \sum_{i_1 i_2 \dots i_{\ell-k}} \left( q_{i_2 i_1}^2 q_{i_3 i_2}^2\dots q_{i_1 i_{\ell-k}}^2\right)^q \left( q_{i_2 i_1}^2 q_{i_3 i_2}^2\dots q_{i_r i_{r-1}}^2\right)} \\
& \stackrel{(b)}{\leq }  \sqrt{ \sum_{i_1 i_2 \dots i_{\ell-k}} \left( q_{i_2 i_1}^2 q_{i_3 i_2}^2\dots q_{i_1 i_{\ell-k}}^2\right)^q \left( q_{i_2 i_1}^2 q_{i_3 i_2}^2\dots q_{i_r i_{r-1}}^2\right)} \\
& \stackrel{(c)}{\leq } \sqrt{ \sum_{i_1 i_2 \dots i_{2\ell-k}} q_{i_2 i_1}^2 q_{i_3 i_2}^2\dots q_{i_{2\ell-k} i_{2\ell-k-1}}^2} \\
& \stackrel{(d)}{\leq } \sqrt{ \gamma \left(\rho_{\max}(\bar Q)\right)^{2\ell - k}}\\
&= \sqrt{\gamma} \left(\rho_{\max}(\bar Q)\right)^{\ell - {k\over 2}} \\
& \stackrel{(e)}{\leq }  \sqrt{\gamma} \left(\rho_{\max}(\bar Q)\right)^{\ell \over 2},
\end{align*}
where $(a)$ follows from the Cauchy-Schwarz inequality and $(b)$ follows from the fact that $\sum_{i} P_i^2 \leq 1$. In $(c)$, some extra terms are added to the inequality. $(d)$ comes from (\ref{eq:markov-inequality}) and finally $(e)$ comes from the fact that $k < \ell$ and $\rho_{\max}(\bar Q)\leq 1$.
\end{proof}

%
%
%
%
%

\section{Proof of Theorem  \ref{p:general_IIDachievability}}\label{sec:noisy-reads}
As explained in Section \ref{sec:formulation-noisy}, the criterion for overlap scoring is based the MAP rule for deciding between two hypotheses: $H_0$ and $H_1$. The null hypothesis $H_0$ indicates that two reads are from same physical source subsequence. Formally, we say two reads $\bR_i$ and $\bR_j$ have  the overlap score $W_{ij}=w$ if $w$ is the longest suffix of $\bR_i$ and prefix of $\bR_j$  passing the criterion (\ref{eq:criterion for decision}).

\def\sanov {f(\ell)}
\def\san#1 {f(#1)}

Let $\sanov=\left(1+\ell \right)^{|\mathcal{X}|}$, where $|\mathcal{X} |$ is the cardinality of the channel's output symbols. The following theorem is a standard result in the hypothesis testing problem, c.f. Chapter 11.7 of \cite{CT06}.

\begin{theorem}\label{thm:hypothesis}
Let $\bX$ and $\bY$ be two random sequences of length $\ell$. For the  given hypotheses $H_0$ and $H_1$ and their corresponding MAP rule  (\ref{eq:criterion for decision}),
$$\P(H_0 | H_1) \leq \sanov 2^{-\ell D(P_{\mu} ||P_X \cdot P_Y)} $$ and
$$\P(H_1 | H_0) \leq \sanov 2^{-\ell D(P_{\mu} ||P_{X,Y})},$$
where
$$P_\mu(x,y) := \frac{[P_{X,Y}(x,y)]^{\mu}[P_X(x)P_Y(y)]^{1-\mu}}{\sum_{a,b} [P_{X,Y}(a,b)]^{\mu}[P_X(a)P_Y(b)]^{1-\mu}}$$ and $\mu$ is the solution of
$$ D(P_{\mu} ||P_X \cdot P_Y)-D(P_{\mu} ||P_{X,Y})=\theta.$$
\end{theorem}

Parallel to the proof of the noiseless case, we first prove the following lemma concerning erroneous merging due to impostor reads.

\begin{lemma}[False alarm]\label{l:false_alarm}
For any distinct $\ell$-mers $\bX$ and $\bY$ from the set of reads:
\begin{enumerate}
\item If the two $\ell$-mers have no physical overlap, the probability that $H_0$ is accepted is 
\begin{equation}
\sanov 2^{-\ell D(P_{\mu} ||P_X \cdot P_Y)}.
\end{equation}
\item If the two $\ell$-mers have physical overlap, the probability that $H_0$ is accepted is
\begin{equation}\label{eq:false_positive_cor}
\gamma \sanov 2^{ - \ell D(P_{\mu} ||P_X \cdot P_Y) /2},
\end{equation} 
where $\gamma$ is a constant.
\end{enumerate}
\end{lemma}
\begin{proof}
The proof of the first statement is an immediate consequence of Theorem \ref{thm:hypothesis}.

We now turn to the second statement. We only consider $\ell =2k$ and the other case can be deduced easily by following similar steps. Let $\chi_j= \log \frac{P(x_j,y_j)}{P(x_j)P(y_j)}$. Since $\chi_j$'s are not independent, we cannot directly use Theorem \ref{thm:hypothesis} to compute $\P \left( \sum_{j=1}^{\ell} \chi_j  \geq \ell \theta \right)$. However, we claim that $\chi_j$'s can be partitioned into two disjoint sets $J_1$ and $J_2$ of the same size, where the $\chi_j$'s within each set are independent. Assuming the claim,
\begin{align*}
 \P \left( \sum_{j=1}^{\ell} \chi_j  \geq \ell \theta \right)
 & = \P \left( \sum_{j\in J_1} \chi_j + \sum_{j\in J_1} \chi_j  \geq \ell \theta \right)  \\
& \stackrel{(a)}{\leq}  \P \left( \sum_{j\in J_1} \chi_j  \geq {\ell \over 2} \theta \right)+ \P \left( \sum_{j\in J_2} \chi_j  \geq {\ell \over 2} \theta \right) \\
& \leq 2  \P \left( \sum_{j\in J_1} \chi_j  \geq {\ell \over 2} \theta \right),
\end{align*}
where $(a)$ follows from the union bound. Since $|J_1|={\ell \over 2}$, one can use Theorem \ref{thm:hypothesis} to show (\ref{eq:false_positive_cor}).

It remains to prove the claim. To this end, let $k$ be the amount of physical overlap between $\bX$ and $\bY$. Without loss of generality, we assume that $S_1S_2\dots S_{2\ell +k}$ is the shared DNA sequence. Let $q$ and $r$ be the quotient and remainder of $\ell$ divided by $2(\ell -k)$, i.e. $\ell = 2q(\ell -k) +r$ where $0\leq r < 2(\ell -k)$. Since $\ell$ is even, $r$ is even.  Let $J_1$ be the set of indices $j$ where either $(j \mod 2(\ell-k))\in \{0,1,\dots,\ell -k-1\}$ for $j \in \{1,\dots, 2q(\ell-k)\}$ or $j \in \{2q(\ell-k)+1,\dots, 2q(\ell-k)+{\ell \over 2} \}$. We claim that the random variables $\chi_j$'s with $j\in J_1$ are independent. We observe that $\chi_j$ depends only on $s_j$ and $s_{j+ (\ell -k)}$. Consider two indices $j_1 < j_2 \in J_1$. The pairs $(s_{j_1},s_{j_1+ (\ell -k)})$ and $(s_{j_2},s_{j_2+ (\ell -k)})$ are disjoint iff $j_1+ (\ell -k) \neq j_2$. By the construction of $J_1$, one can show that $j_1+ (\ell -k) \neq j_2$ for any $j_1 < j_2 \in J_1$. Hence, $\chi_j$'s with $j\in J_1$ are independent. A similar argument shows $\chi_j$'s with $j\in J_2=\{1,\dots,2\ell-k\}- J_1$ are independent. This completes the proof.
\end{proof}

Due to noise,  two physically overlapping reads may not pass the criterion. To deal with this event, we state the following lemma.

\begin{lemma}[Mis-detection]\label{l:mis_detection}
Let $\bX$ and $\bY$ be two distinct  $\ell$-mers  from the same physical location. The probability that $H_1$ is accepted is bounded by
$$\sanov 2^{-\ell D(P_{\mu} ||P_{X,Y})}.$$
\end{lemma}
\begin{proof}
This is an immediate consequence of Theorem \ref{thm:hypothesis}.
\end{proof}

{\bf Proof of Theorems  \ref{p:general_IIDachievability}:} Similar to the proof of achievability result in the noiseless case, we decompose the error event $\EE$ into $\EE_1\cup \EE_2$ where $\EE_1$ is the event that some read is merged incorrectly and $\EE_2$ is the event that the DNA sequence is not covered by the reads. The probability of the second event, similar to the noiseless case, goes to zero exponentially fast if $R> \bar L$. We only need to compute $\P(\EE_1)$. Again, $\EE_1$ can be decomposed as $\EE_1=\cup_\ell \calA_\ell$, where $\calA_\ell$ is the event that the first error in merging occurs at stage $\ell$. Moreover, 
\begin{equation}\label{e:AdecombN}
\calA_\ell \subseteq \BB_{\ell}\cup \CC_{\ell}\,,
\end{equation}
where:
\begin{eqnarray}
\BB_\ell & \triangleq & \{\bR_j\neq \bR_i^s, U_j\leq \ell, V_i\leq \ell, W_{ij}=\ell \mbox{ for some $i \not = j$.}\} \label{eq:eventBN}\\
\CC_\ell & \triangleq & \{\bR_j=\bR_i^s, U_j=V_i\neq  \ell, W_{ij}=\ell \mbox{ for some $i \not = j$.}\} \label{eq:eventCN}
\end{eqnarray}
Note that here the definition of $\CC_\ell$ is different from that of \eqref{eq:eventC} as for the noiseless reads the overlap score is never less than the physical overlap. However, in the noisy reads there is a chance for observing this event due to mis-detection.

The analysis of $\BB_\ell$ follows closely from that of the noiseless case. In fact, using Lemma \ref{l:false_alarm} which is a counterpart of Lemma \ref{lem self match iid} and following similar steps in calculation of $\P(\BB_\ell)$ in the noiseless case, one can obtain
\begin{align}
\P(\BB_\ell)  \leq &  \sanov ( q_\ell ^2 + 2\gamma \lambda  L q_\ell ),
\label{eq:pbN}
\end{align}
where 
\begin{align}
q_\ell = \lambda G  e^{-\lambda (L-\ell)(1-2/N)} 2^{ -\ell D(P_{\mu} ||P_X \cdot P_Y)/2}.
\end{align}

To compute $\P(\CC_\ell)$, we note that $\CC_\ell \subseteq \cup_{i} \CC_\ell ^{i}$, where 
$$ \CC_\ell ^i  \triangleq  \{V_i= \ell, W(\bR_i, \bR_i^s)\neq \ell \}. $$
Applying the union bound and considering the fact that $\CC_\ell ^{i}$'s are equiprobable yields
$$\CC_\ell \leq N \P(\CC_\ell ^1).$$
Hence,
\begin{align*}
\P(\CC_\ell) & \leq N (\P( W(\bR_i, \bR_i^s)> \ell  | V_i = \ell)+\P( W(\bR_i, \bR_i^s) < \ell  | V_i = \ell)) \P(V_i< \ell).
\end{align*}
Using Lemma \ref{l:false_alarm} part 2 and Lemma \ref{l:mis_detection} yields
\begin{align*}
\P(\CC_\ell) & \leq \lambda G \sanov \left(\gamma e^{ - \ell D(P_{\mu} ||P_X \cdot P_Y) /2}+ e^{ -\ell D(P_{\mu} ||P_{X,Y})}\right)   e^{-\lambda(L-\ell)(1-1/N)}\\
&  \leq  \sanov \left(\gamma q_\ell + q'_\ell \right),
\end{align*}
where
$$q'_\ell =\lambda G e^{ -\ell D(P_{\mu} ||P_{X,Y})} 2^{-\lambda(L-\ell)(1-1/N)}.$$

Combining all the terms, we obtain
\begin{align*}
\P(\EE_1)  \leq  \sum_{\ell = 0}^L \P(\BB_\ell) + \P(\CC_\ell) \leq \sum_{\ell=0}^{L}  \sanov ( q_\ell ^2 +  \gamma (2 \lambda L +1) q_\ell + q'_\ell  ).
 \end{align*}
 
To show that that $\P(\EE_1) \rightarrow 0$, it is sufficient to argue that $q_0$, $q'_0$, $q_L$,  and $q'_L$  go to zero exponentially in $L$. Considering first $q_0$ and $q'_0$, they vanish exponentially in $L$ if $N > G\ln G/ L$ which implies $\capacity(\bar L)=1 \,$. The terms $q_L$ and $q'_L$ vanish exponentially in $L$ if 
$$\bar{L}>\frac{2}{\min (2D(P_{\mu} ||P_{X,Y}),D(P_{\mu} ||P_X \cdot P_Y))}.$$
%
%

Since  $\P(\EE_1)=o(1)$ and $\P(\EE_2)=o(1)$ for any choice of $\theta$, one can optimize over $\theta$ to obtain the result given in the theorem. This completes the proof.

\bibliographystyle{amsplain}
\bibliography{references}

\end{document}